\documentclass[11pt,final]{article}
\usepackage{microtype}%if unwanted, comment out or use option "draft"
\usepackage{xspace}
\usepackage{bbm}
\usepackage{graphicx}
\usepackage{amsmath,amssymb,amsthm}
\usepackage{tikz}
\usepackage{thm-restate}
\usepackage[normalem]{ulem}

\usepackage{paralist}
\usepackage{bm}
\usepackage{xspace}
\usepackage{url}
\usepackage{fullpage, prettyref}
\usepackage{boxedminipage,wrapfig}
\usepackage{wrapfig}
\usepackage{ifthen}
\usepackage{color}
\usepackage{xcolor}
\usepackage{algpseudocode}
\usepackage{caption}
\usepackage{algorithm}
\usepackage{fullpage}
\usepackage{hyperref}
\usepackage[compact]{titlesec}
\usepackage{cleveref}
\usepackage{nicefrac}

\usepackage[bottom=1in,top=1in,left=1in,right=1in]{geometry}

\usepackage{framed}
\renewenvironment{leftbar}[1][\hsize]
{%
    \MakeFramed{\hsize#1\advance\hsize-\width\FrameRestore}%
}
{\endMakeFramed}

\usepackage{thmtools}
\usepackage{thm-restate}
\declaretheorem[numberwithin=section]{theorem}
\declaretheorem[sibling=theorem]{lemma}
\declaretheorem[sibling=theorem]{claim}

\declaretheorem[sibling=theorem]{corollary}

\newtheorem{definition}[theorem]{Definition}
\newtheorem{observation}[theorem]{Observation}
\crefname{claim}{claim}{claims}

\usepackage{framed}
\renewenvironment{leftbar}[1][\hsize]
{%
    \MakeFramed{\hsize#1\advance\hsize-\width\FrameRestore}%
}
{\endMakeFramed}

\newenvironment{subproof}[1][\proofname]{%
  \begin{proof}[#1]%
}{%
  \end{proof}%
}

\usepackage[colorinlistoftodos,textsize=tiny,textwidth=2cm,color=red!25!white,obeyFinal]{todonotes}
\newcommand{\agnote}[1]{\todo[color=blue!25!white]{AG: #1}\xspace}
\newcommand{\elnote}[1]{\todo[color=green!25!white]{EL: #1}\xspace}

\makeatletter
\algrenewcommand\algorithmiccomment[2][\normalsize]{{#1\hfill\(\triangleright\) \emph{#2}}}
\makeatother

\newcommand{\ignore}[1]{}

\newcommand{\f}{\displaystyle\frac}
\newcommand{\cd}{\cdot}
\newcommand{\bn}{\binom}
\newcommand{\lds}{\ldots}
\newcommand{\s}{\subseteq}
\newcommand{\inv}{^{-1}}
\newcommand{\al}{\alpha}
\newcommand{\Om}{\Omega}
\newcommand{\el}{\ell}
\newcommand{\m}{\mathcal}

\newcommand{\lf}{\lfloor}
\newcommand{\rf}{\rfloor}
\newcommand{\lc}{\lceil}
\newcommand{\rc}{\rceil}

\newcommand{\1}{\mathbbm 1}
\newcommand{\lp}{\left(}
\newcommand{\rp}{\right)}
\newcommand{\lb}{\left[}
\newcommand{\rb}{\right]}

\newcommand{\eps}{\varepsilon}
\newcommand{\e}{\varepsilon}
\newcommand{\om}{\omega}

\newcommand{\poly}{\mathrm{poly}}

\newcommand{\calP}{{\cal P}}

\newcommand{\calS}{{\cal S}}

\newcommand{\initOneLiners}{%
    \setlength{\itemsep}{0pt}
    \setlength{\parsep }{0pt}
    \setlength{\topsep }{0pt}
}
\newenvironment{OneLiners}[1][\ensuremath{\bullet}]
    {\begin{list}
        {#1}
        {\initOneLiners}}
    {\end{list}}

\newcommand{\Opt}{\ensuremath{\mathsf{Opt}}\xspace}
\newcommand{\sse}{\subseteq}

\usepackage{caption}

\DeclareCaptionFormat{algor}{%
  \hrulefill\par\offinterlineskip\vskip1pt%
    \textbf{#1#2}#3\offinterlineskip\hrulefill}
\DeclareCaptionStyle{algori}{singlelinecheck=off,format=algor,labelsep=space}
\newcommand{\kcut}{$k$\textsc{-Cut}\xspace}

\newcommand{\Laminarkcut}[1]{\textsc{Laminar }$#1$\textsc{-cut}} 
 
\newcommand{\defn}[1]{\textbf{#1}}

\newcommand{\tO}{\widetilde{O}}
\newcommand{\kclique}{$k$\textsc{-Clique}\xspace}
\newcommand{\cut}{{\small\textsf{Cut}}}

\newcommand{\Ttree}{\text{T-tree}\xspace}

\newcommand{\kl}{\mathfrak{a}}
\newcommand{\kr}{\mathfrak{b}}

\newcommand{\pvc}{\textsc{Partial VC}\xspace}

\newcommand{\desc}{\ensuremath{\mathrm{desc}}}

\newcommand{\apxfactor}{1.81}

\algnewcommand{\IIf}[1]{\State\algorithmicif\ #1\ \algorithmicthen}
\algnewcommand{\EndIIf}{\unskip\ \algorithmicend\ \algorithmicif}

\begin{document}

\title{{\bf Faster Exact and Approximate Algorithms for $k$-Cut}}

\author{ Anupam Gupta\thanks{{\tt anupamg@cs.cmu.edu}. Supported in part by NSF awards
    CCF-1536002, CCF-1540541, and CCF-1617790. } \\ CMU \and Euiwoong Lee\thanks{{\tt euiwoong@cims.nyu.edu}. Part of this work was done as a research fellow at the Simons Institute. }\\ NYU
  \and Jason Li\thanks{{\tt jmli@cs.cmu.edu}. Supported in part by NSF awards
    CCF-1536002, CCF-1540541, and CCF-1617790. }\\ CMU}

\date{}

\thispagestyle{empty}
\maketitle
\begin{abstract}
  In the \kcut problem, we are given an edge-weighted graph $G$ and an
  integer $k$, and have to remove a set of edges with minimum total
  weight so that $G$ has at least $k$ connected components. The current best
  algorithms are an $O(n^{(2-o(1))k})$ randomized algorithm due
  to Karger and Stein, and an $\smash{\tilde{O}}(n^{2k})$
  deterministic algorithm due to Thorup. Moreover,
  several $2$-approximation algorithms are known for the problem (due to
  Saran and Vazirani, Naor and Rabani, and Ravi and Sinha). 

  It has remained an open problem to (a) improve the runtime of exact
  algorithms, and (b) to get better approximation algorithms. In this
  paper we show an $O(k^{O(k)} \, n^{(2\omega/3 + o(1))k})$-time
  algorithm for \kcut. Moreover, we show an $(1+\e)$-approximation
  algorithm that runs in time $O((k/\e)^{O(k)} \,n^{k + O(1)})$, and a
  $\apxfactor$-approximation in fixed-parameter time $O(2^{O(k^2)}\,\poly(n))$.

\end{abstract}

\newpage

\setcounter{page}{1}

\section{Introduction}
\label{sec:introduction}

In this paper we consider the \kcut problem: given an edge-weighted
graph $G = (V,E,w)$ and an integer $k$, delete a minimum-weight set of
edges so that $G$ has at least $k$ connected components. This problem is
a natural generalization of the global min-cut problem, where the goal
is to break the graph into $k=2$ pieces. 
This problem has been actively studied in theory of both exact and
approximation algorithms, where each result brought new insights and tools
on graph cuts. %\elnote{Added last sentence.}

It is not \emph{a priori} clear how to obtain poly-time algorithms for
any constant $k$, since guessing one vertex from each part only
reduces the problem to the NP-hard \textsc{Multiway Cut}
problem. Indeed, the first result along these lines was the work of
Goldschmidt and Hochbaum~\cite{GH94} who gave an
$O(n^{(1/2 - o(1))k^2})$-time exact algorithm for \kcut. 
Since then, the exact exponent in terms of $k$ 
has been actively studied. %\elnote{Added this sentence.}
The current
best runtime is achieved by an $\tO(n^{2(k-1)})$ randomized algorithm
due to Karger and Stein~\cite{KS96}, which performs random edge
contractions until the remaining graph has $k$ nodes, and shows that the
resulting cut is optimal with probability at least
%$k \binom{n}{k-1}^{-1} \binom{n-1}{k-1}^{-1} = \Omega(n^{-2(k-1)})$. 
$\Omega(n^{-2(k-1)})$.  The asymptotic runtime of
$\smash{\tilde{O}}(n^{2(k-1)})$ was later matched by a deterministic
algorithm of Thorup~\cite{Thorup08}. His algorithm was based on
tree-packing theorems; it showed how to efficiently find a tree for
which the optimal $k$-cut crosses it $2k-2$ times. Enumerating over all
possible $2k-2$ edges of this tree gives the algorithm. 

These elegant $O(n^{2k})$-time algorithms are the state-of-the-art, and
it has remained an open question to improve on them. An easy observation
is that the problem is closely related to \kclique, so we may not expect
the exponent of $n$ to go below $(\om/3)k$. Given the interest in
fine-grained analysis of algorithms, where in the range $[(\om/3)k, 2k-2]$
does the correct answer lie?

\iffalse
\emph{Though there was no
published time lower bound (we observe one below), the elegance of both
Karger and Stein's and Thorup's algorithms and their natural connections
to graph properties made it hard to believe (at least for the authors)
that one can significantly improve their running times.} \elnote{Not sure
  about this sentence..} \agnote{Me too. Too bold.... :) }
\fi

Our main results give faster deterministic and randomized algorithms for
the problem. %(We use \whp to denote probability at least $1 - 1/\poly(n)$.)

\begin{restatable}[Faster Randomized Algorithm]{theorem}{FasterKCut}
  \label{thm:FasterKCut}
  Let $W$ be a positive integer. There is a randomized algorithm for
  \kcut on graphs with edge weights in $[W]$ with runtime
  \[ \tO(k^{O(k)}n^{k+\lf(k-2)/3\rf\om +1+ ((k-2)\bmod3)}W) \approx
    O(k^{O(k)} n^{(1+ \om/3)k}),\] that succeeds with probability
  $1 - 1/\poly(n)$.
\end{restatable}

\begin{restatable}[Even Faster Deterministic
  Algorithm]{theorem}{EvenFasterKCut}
  \label{thm:EvenFasterKCut}
  Let $W$ be a positive integer. For any $\e>0$, there is a
  deterministic algorithm for exact \kcut on graphs with edge weights in
  $[W]$ with runtime
  \[ k^{O(k)}n^{(2\om/3+\e)k+O(1)}W \approx
    O(k^{O(k)} n^{(2\om/3)k}).\]
\end{restatable}
%\elnote{$\tilde{O}, \approx$?}
In the above theorems, $\om$ is the matrix multiplication constant, and
$\tO$ hides polylogarithmic terms. %\agnote{Hope that $\approx$ is OK.}
While the deterministic algorithm from Theorem~\ref{thm:EvenFasterKCut}
is asymptotically faster, the randomized algorithm is better for small
values of $k$. Indeed, using the current best value of $\om<2.373$~\cite{le2014powers},
Theorem~\ref{thm:FasterKCut} gives a randomized algorithm for exact
\kcut on unweighted graphs % in time
% $\tO(k^{O(k)}n^{k+\lf(k-1)/3\rf\om + ((k-1)\mod3)})$. This 
which improves upon the previous best $\tO(n^{2k-2})$-time algorithm of
Karger and Stein for all $k \in [8,n^{o(1)}]$. For $k\le6$, faster
algorithms were given by Levine~\cite{Levine00}. %\agnote{Jason, could
%  you fix?}

% \begin{theorem}[Main Theorem]
%   \label{thm:kcut-main}
%   Let $W$ be a positive integer. For any $\e>0$, there is a randomized
%   algorithm for \kcut that runs in expected
%   $k^{O(k)}n^{(2\om/3+\eps)k+O(1)}W$ time on graphs with edge weights in
%   $[W]$. Here $\om\le 2.373$ is the matrix multiplication constant.
% \end{theorem}

% \alert{Say something about the space-efficient version here. Not sure
%   yet about putting it as a theorem.}

\paragraph{Approximation algorithms.}
The \kcut problem has also received significant attention from the
approximation algorithms perspective. There are several
$2(1-1/k)$-approximation algorithms that run in time $\poly(n,
k)$~\cite{SV95, NR01, RS02}, which cannot be improved assuming the Small
Set Expansion Hypothesis~\cite{Manurangsi17}.  Recently, we gave an
$1.9997$-approximation algorithm that runs in $2^{O(k^6)}
n^{O(1)}$~\cite{GuptaLL18}.
%\elnote{How to say us?}
% \agnote{Maybe simpler
%   is better?}.  
In this current paper, we give a $(1+\e)$-approximation
algorithm for this problem much faster than the current best exact
algorithms; prior to our work, nothing better was known for
$(1+\e)$-approximation than for exact solutions.

\begin{restatable}[Approximation]{theorem}{ApproxKCut}
  \label{thm:kcut-approx}
  For any $\e>0$, there is a randomized (combinatorial) algorithm for
  \kcut with runtime $(k/\e)^{O(k)} n^{k+O(1)}$ time on general graphs,
  that outputs a $(1+\e)$-approximate solution with probability $1 -
  1/\poly(n)$.
\end{restatable}

The techniques from the above theorem, combined with the previous 
ideas in~\cite{GuptaLL18}, immediately give an improved FPT 
approximation guarantees for the
\kcut problem:

\begin{theorem}[FPT Approximation]
  \label{thm:fpt-apx}
  There is a deterministic $\apxfactor$-approximation algorithm for the
  \kcut problem that runs in time $2^{O(k^2)} \cdot
  n^{O(1)}$. %\agnote{Check and make this more precise.}
\end{theorem}

\paragraph{Limitations.}
Our exact algorithms raise the natural question: how fast can exact
algorithms for \kcut be? We give a simple reduction showing that while
there is still room for improvement in the running time of exact algorithms, such improvements can only improve the constant in front of
the $k$ in the exponent, assuming a popular conjecture on
algorithms for the \textsc{Clique} problem.

\begin{restatable}[Relationship to Clique]{claim}{Hardness}
%\begin{fact}[Hardness]
  \label{fact:hard}
  Any exact algorithm for the \kcut problem for graphs with edge weights
  in $[n^2]$ can solve the $k$-\textsc{Clique} problem in the
  same runtime. Hence, assuming $k$-\textsc{Clique} cannot be solved in
  faster than $n^{\omega k/3}$ time, the same lower bound holds for the
  $k$-cut problem.
\end{restatable}

% Moreover, we show that the analysis of the Karger-Stein algorithm is
% tight. I.e., it does not give an algorithm faster than $\approx
% n^{2k-2}$. \agnote{Write this proof. Also, can we show that Thorup is
%   no better than claimed?}

% \begin{theorem}[Lower-Bound for Karger-Stein]
%   For every $n$, there are unweighted graphs on which the Karger-Stein
%   algorithm finds the optimal $k$-cut with probability at most
%   $k^{O(k)}n^{-2k}$.
% \end{theorem}

\subsection{Our Techniques}

Our algorithms build on the approach pioneered by Thorup: using
tree-packings, he showed how to find a tree $T$ such that it crosses the
optimal $k$-cut at most $2k-2$ times. (We call such a tree a
\emph{Thorup tree}, or \emph{\Ttree}.)  Now brute-force search over
which edges to delete from the \Ttree (and how to combine the resulting
parts together) gave an $\tO(n^{2k-2})$-time deterministic
algorithm. This last step, however, raises the natural
question---\emph{having found such a \Ttree, can we use the structure of
  the \kcut problem to beat brute force?} Our algorithms answer the
question in the affirmative, in several different ways. The main ideas
behind our algorithm are dynamic programming and fast
matrix-multiplication, carefully combined with the fixed-parameter
tractable algorithm technique of color-coding, and random sampling in
general.

\paragraph{Fast matrix multiplication.}
Our idea to apply fast matrix multiplication 
starts with the crucial observation that if
(i)~the \Ttree $T$ is ``tight'' and crosses the optimal $k$-cut only
$k-1$ times, and (ii)~these edges are ``incomparable'' and do not lie on
a root-leaf path, then the problem of finding these $k-1$ edges can be
modeled as a max-weight clique-like problem! (And hence we can use
matrix-multiplication ideas to speed up their computation.) 
An important property of this special case is that choosing an edge $e$ to cut
fixes one component in the \kcut solution --- by incomparability,
the subtree below $e$ cannot be cut anymore. 
The cost of a $k$-cut can be determined by the weight of edges between
each pair of components (just like being a clique is determined by pairwise connectivity), 
so this case can be solved via an algorithm similar to \kclique. 

\paragraph{Randomized algorithm.}
Our randomized algorithm removes these two assumptions step by step.
First, while the above intuition crucially relies on assumption~(ii), we
give a more sophisticated dynamic program using color-coding schemes for
the case where the edges are not incomparable. %\elnote{Highlight color
%  coding here?}  
Moreover, to remove assumption~(i), we show a
randomized reduction that given a tree that crosses the optimal cut as
many as $2k-2$ times, finds a ``tight'' tree with only $k-1$ crossings
(which is the least possible), at the expense of a runtime of
$O(k^2n)^{k-1}$. Note that guessing which edges to delete is easily done
in $n^{k-1}$ time, but adding edges to regain connectivity while not
increasing the number of crossings can naively take a factor of
$m^{k-1}$ more time. We lose only a $k^{2(k-1)}$ factor using our
random-sampling based algorithm, using that in an optimal \kcut a
split cluster should have more edges going to its own parts than to
other clusters.

\paragraph{Deterministic algorithm.}
The deterministic algorithm proceeds along a different direction and
removes both assumptions~(i) and~(ii) at once.  We show that by deleting
some $O(\log k)$ carefully chosen edges from the \Ttree $T$, we can
break it into three forests such that we only need to delete about
$2k/3$ edges from each of these forests.  Such a deletion is not
possible when $T$ is a star, but appropriately extending $T$ by
introducing Steiner nodes admits this deletion. (And $\Theta(\log k)$ is
tight in this extension.)  For each forest, there are $n^{2k/3}$ ways to
cut these edges, and once a choice of $2k/3$ edges is made, the forest
will not be cut anymore.  This property allows us to bypass (ii) and
establish desired pairwise relationships between choices to delete
$2k/3$ edges in two forests. Indeed, we set up a tripartite graph where
one part corresponds to the choices of which $\leq 2k/3$ edges to cut in
one forest and the cost of the min $k$-cut is the weight of the
min-weight triangle, which we find efficiently using fast matrix
multiplication.  Some technical challenges arise because we need to
%merge those $2k$ components back into $k$ components and 
some components for some forests may only have Steiner vertices, but we
overcome these problems using color-coding.
%We want triangles in this graph to give us solutions, but we need
%these solutions to be consistent. \alert{say more?} For this, we use a
%color-coding idea of guessing which components intersect which of the
%forests, and using this information to add edges. This again reduces to
%a max-weight triangle problem, whence we get our improvement.

\paragraph{Approximation schemes.}
The $(1+\e)$-approximation algorithm again uses the $O(k^2n)^{k-1}$-time
randomized reduction, so that we have to cut exactly $k-1$ edges from
a ``tight'' \Ttree $T$. An exact dynamic program for this problem takes
time $\approx n^k$ --- as it should, since even this tight case captures
clique, when $T$ is a star and hence these $k-1$ edges are incomparable.
And again, we need to handle the case where these $k-1$ edges are not
incomparable. For the former problem, we replace the problem of finding
cliques by approximately finding ``partial vertex covers'' instead. (In
this new problem we find a set of $k-1$ vertices that minimize the total
number of edges incident to them.) Secondly, in the DP we cannot afford
to maintain the ``boundary'' of up to $k$ edges explicitly any more. We
show how to maintain an ``$\eps$-net'' of nodes so that carefully
``rounding'' the DP table to only track a small $f(k)$-sized set of
these rounded subproblems incurs only a $(1+\eps)$-factor loss in
quality.

Our approximate DP technique turns out to be useful to get a
$\apxfactor$-approximation for \kcut in FPT time, improving on our
previous approximation of $\approx 1.9997$~\cite{GuptaLL18}. In particular,
the \emph{laminar cut problem} from~\cite{GuptaLL18} also has a tight \Ttree
structure, and hence we can use (a special case of) our approximate DP
algorithm to get a $(1+\e)$-approximation for laminar cut, instead of
the $2-\e$-factor previously known. Combining with other ideas in the
previous paper, this gives us the $\apxfactor$-approximation.

% Hence we give an approximation based on two
% ideas: the use of the subroutine, to 

\subsection{Related Work}

The first non-trivial exact algorithm for the \kcut problem was by
Goldschmidt and Hochbaum, who gave an $O(n^{(1/2- o(1))k^2})$-time
algorithm~\cite{GH94}; this is somewhat surprising because the related
\textsc{Multiway Cut} problem is NP-hard even for $k=3$. They also
proved the problem to be NP-hard when $k$ is part of the input.  Karger
and Stein improved this to an $O(n^{(2-o(1))k})$-time randomized
Monte-Carlo algorithm using the idea of random
edge-contractions~\cite{KS96}. Thorup improved the $O(n^{4k +
  o(1)})$-time deterministic algorithm of Kamidoi et al.~\cite{KYN06} to
an $\tilde{O}(n^{2k})$-time deterministic algorithm based on tree
packings~\cite{Thorup08}. Better algorithms are known for small values
of $k \in [2, 6]$~\cite{NI92, HO92, BG97, Karger00, NI00, NKI00,
  Levine00}.

\paragraph{Approximation algorithms.} The first such result for \kcut was a
$2(1-1/k)$-approximation of Saran and Vazirani~\cite{SV95}.  Later, Naor
and Rabani~\cite{NR01}, and also Ravi and Sinha~\cite{RS02} gave
$2$-approximation algorithms using tree packing and network strength
respectively.  Xiao et al.~\cite{XCY11} extended Kapoor~\cite{Kapoor96}
and Zhao et al.~\cite{ZNI01} and generalized Saran and Vazirani to give
an $(2 - h/k)$-approximation in time $n^{O(h)}$. On the hardness front,
Manurangsi~\cite{Manurangsi17} showed that for any $\eps > 0$, it is
NP-hard to achieve a $(2 - \eps)$-approximation algorithm in time
$\poly(n,k)$ assuming the Small Set Expansion Hypothesis.

In recent work~\cite{GuptaLL18}, we gave a $1.9997$-approximation for
\kcut in FPT time $f(k) \poly(n)$; this does not contradict Manurangsi's
work, since $k$ is polynomial in $n$ for his hard instances. We improve
that guarantee to $\apxfactor$ by getting a better approximation ratio
for the ``laminar'' $k$-cut subroutine, improving from $2-\e$ to $1+\e$.
This follows as a special case of the techniques we develop in
\S\ref{sec:approx}; the rest of the ideas in this current paper are
orthogonal to those in~\cite{GuptaLL18}.

\paragraph{FPT algorithms.} Kawarabayashi and Thorup give the first $f(\Opt)
\cdot n^{2}$-time algorithm~\cite{KT11} for unweighted graphs. Chitnis
et al.~\cite{Chitnis} used a randomized color-coding idea to give a
better runtime, and to extend the algorithm to weighted graphs. Here,
the FPT algorithm is parameterized by the cardinality of edges in the
optimal \kcut, not by the number of parts $k$. For more details on FPT
algorithms and approximations, see the book~\cite{FPT-book}, and the 
survey~\cite{Marx07}.

\subsection{Preliminaries}
\label{sec:preliminaries}

For a graph $G=(V,E,w)$, consider some collection of disjoint sets
$\calS = \{S_1,\lds,S_r\}$. Let $E_G(\calS) = E_G(S_1,\lds,S_r)$ denote
the set of edges in $E_G[\cup_{i = 1}^r S_r]$ (i.e., among the edges
both of whose endpoints lie in these sets) whose endpoints belong to
different sets $S_i$. For any vertex set $S$, let $\partial S$ denote
the edges with exactly one endpoint in $S$; hence
$E_G(\calS) = \cup_{S_i \in \calP} \, \partial S_i$.  For a collection
of edges $F\s E$, let $w(F):=\sum_{e\in F}w(e)$ be the sum of weights of
edges in $F$. In particular, for a \kcut solution $\{S_1,\lds,S_k\}$,
the value of the solution is $w(E_G(S_1,\lds,S_k))$.

For a rooted tree $T=(V_T,E_T)$, let $T_v\s V_T$ denote the subtree of
$T$ rooted at $v\in V_T$. For an edge $e\in E_T$ with child vertex $v$,
let $T_e:=T_v$. Finally, for any set $S \s V_T$, $T_S = \sum_{v \in S} T_v$.

For some sections, we make no assumptions on the edge weights of $G$,
while in other sections, we will assume that all edge weights in $G$ are
integers in $[W]$, for a fixed positive integer $W$. We default to the
former unrestricted case, and explicitly mention transitioning to the
latter case when needed.

%%% Local Variables:
%%% mode: latex
%%% TeX-master: "main"
%%% End:

%\newpage
\section{A Fast Randomized Algorithm}
\label{sec:randomized}

In this section, we present a randomized algorithm to solve \kcut exactly
in time $\tO(k^{O(k)} n^{(1+ \omega/3)k})$, proving 
Theorem~\ref{thm:FasterKCut}. 
Section~\ref{sec:Thorup} introduces our high-level ideas based on
Thorup's tree packing results.
Section~\ref{sec:Merge} shows
how to refine Thorup's tree to a good tree that crosses the optimal $k$-cut
exactly $k - 1$ times, and
Section~\ref{sec:Remove} presents an algorithm given a good tree.

\subsection{Thorup's Tree Packing and Thorup's Algorithm}
\label{sec:Thorup}

Our starting point is a transformation from the general \kcut problem to
a problem on trees, inspired by Thorup's algorithm~\cite{Thorup08} based
on greedy tree packings. We will be interested in trees that cross the
optimal partition only a few times.
We fix an optimal \kcut solution, $\m S^*:=\{S^*_1,\lds,S^*_k\}$. Let
$OPT:=E_G(S^*_1,\lds,S^*_k)$ be edges in the solution, so that $w(OPT)$
is the solution value.

\begin{definition}[{\Ttree}s]
  \label{def:ttree}
  A tree $T$ of $G$ is a $\ell$-\Ttree if it crosses the
  optimal cut at most $\ell$ times; i.e., $E_T(S^*_1,\lds,S^*_k) \leq
  \ell$.  If $\ell = 2k-2$, we often drop the quantification and call it
  a \Ttree. If $\ell = k-1$, the minimum value possible, then we call it
  a tight \Ttree.
\end{definition}

Our first step is the same as in~\cite{Thorup08}: we compute a
collection $\m T$ of $n^{O(1)}$ trees such that there exists a \Ttree,
i.e., a tree $T \in \m T$ that crosses $OPT$ at most $2k-2$ times.
% In order to do this, we use the following theorem of . We do not need
% any property of the greedy tree packing process, other than that it
% can be done in polynomial time, so we do not define it here.
\begin{theorem}[\cite{Thorup08}, Theorem~1]
  For $\al\in(0,\frac9{10})$, let $\m T$ be a greedy tree packing with
  at least $3m(k/\al)^3\ln(nmk/\al)$ trees. Then, on the average, the
  trees $T\in\m T$ cross each minimum $k$-cut less than $2(k-1+2\al)$
  times. Furthermore, the greedy tree packing algorithm takes
  $\tO(k^3m^2)$ time.
\end{theorem}

The running time comes from the execution of  $\tO(k^{3}m)$ minimum spanning tree computations. Note that, since our results are only interesting when $k\ge7$, resulting in algorithms of running time $\Om(n^{7+2\om})$, we can completely ignore the running time of the greedy tree packing algorithm, which is only run once.
Letting $\al:=1/8$, we get the following corollary.

\begin{corollary}\label{cor:TreePacking}
  We can find a collection of $\tO(k^3m)$ trees such that for a random
  tree $T\in\m T$, $|E_T(S^*_1,\lds,S^*_k)|\le 2k-3/2$ in
  expectation. In particular, there exists a \Ttree $T\in \m T$.%  that
  % contains at most $2k-2$ edges from the optimal $k$-cut
        %there exists a tree $T\in \m T$ with  $|E_T(S^*_1,\lds,S^*_k)|\le 2k-2$.
\end{corollary}

In other words, if we choose such a \Ttree $T\in \m T$, we get the
following problem: find the best way to cut $\le 2k-2$ edges of $T$, and
then merge the connected components into exactly $k$ components
$S_1,\lds,S_k$ so that $E_G(S_1,\lds,S_k)$ is minimized. Thorup's
algorithm accomplishes this task using brute force: try all possible
$O(n^{2k-2})$ ways to cut and merge, and output the best one. This gives
a runtime of $\tO(k^3n^{2k-2}m)$, or even $\tO(n^{2k-2}m)$ with a more
careful analysis~\cite{Thorup08}. The natural question is: can we do
better than brute-force?

% Since Thorup's brute force approach already produces an $\tO(n^{2k})$ algorithm, the next natural question to ask is: can we do faster than brute force per tree in $\m T$ by, say, an $n^{\Om(k)}$ factor, which would lead to an $O(n^{(2-\e)k})$ algorithm?

For the min-cut problem (when $k=2$), Karger was able to speed up this
step from $O(n^{2k-2})=O(n^2)$ to $\tO(n)$ using dynamic tree data
structures~\cite{Karger00}. However, this case is special:
since there are $\le 3$ components produced
from cutting the $\le 2k-2=2$ tree edges, only one pair of components need to be merged.
For larger values of $k$, it is not clear
how to generalize the use of clever data structures to handle multiple
merges.

Our randomized algorithm gets the improvement in three steps:
\begin{itemize}
\item First, instead of trying all possible trees $T\in\m T$, we only
  look at a random subset of $\Om(k\log n)$ trees. By
  Corollary~\ref{cor:TreePacking} and Markov's inequality, the
  probability that a random tree satisfies
  $|E_T(S_1^*,\lds,S_k^*)|\ge2k-1$ is $\le
  (2k-3/2)/(2k-1)=1-\Om(1/k)$. Therefore, by trying $\Om(k\log n)$
  random trees, we find a \Ttree $T$ % with $|E_T(S_1^*,\lds,S_k^*)|\le2k-2$
  w.h.p.
\item Next, given a \Ttree $T$ from above, we show how to find a collection of
  $\approx n^{k-1}$ trees such that, with high probability, one of these
  trees $T'$ is a tight \Ttree, i.e., it intersects $OPT$ in exactly
  $k-1$ edges. We show this in \S\ref{sec:Merge}.
\item Finally, given a tight \Ttree $T'$ from the previous step, we show
  how to solve the optimal \kcut in time $\approx O(n^{(\omega/3)k})$,
  much like the \kclique problem~\cite{nevsetvril1985complexity}. The
  runtime is not coincidental; the $W[1]$ hardness of \kcut derives from
  \kclique, and hence techniques for the former must work also for the
  latter.  We show this in \S\ref{sec:Remove}.
\end{itemize}

\subsection{A Small Collection of ``Tight'' Trees}\label{sec:Merge}

In this section we show how to find a collection of $\approx n^{k-1}$
trees such that, with high probability, one of these trees $T'$ is a
tight \Ttree. % i.e., it intersects $OPT$ in exactly
% $k-1$ edges.
Formally, 

\begin{restatable}{lemma}{Merge}
  \label{lem:Merge}
  There is an algorithm that takes as input a tree $T$ such that
  $|E_T(S^*_1,\lds,S^*_k)|\le2k-2$, and produces a collection of
  $k^{O(k)}n^{k-1}\log n$ trees, such that one of the new trees $T'$
  satisfies $|E_{T'}(S^*_1,\lds,S^*_k)|=k-1$ w.p.\ $1 - 1/\poly(n)$. The
  algorithm runs in time $k^{O(k)}n^{k-1}m\log n$.
\end{restatable}

%For ease of presentation, we first introduce our algorithm in a completely randomized manner. This is sufficient to obtain the randomized running time of Theorem~\ref{thm:FasterKCut}. 

%\paragraph{The Game Plan.}
The algorithm proceeds by iterations. In each iteration, our goal is to
remove one edge of $T$ and then add another edge back in, so that the
result is still a tree. In doing so, the value of $|E_T(S^*_1,\lds,S^*_k)|$
can either decrease by $1$, stay the same, or increase by $1$. We call
an iteration \textit{successful} if $|E_T(S^*_1,\lds,S^*_k)|$ decreases by
$1$. Throughout the iterations, we will always refer to $T$ as the
current tree, which may be different from the original tree. Finally, if
$|E_T(S^*_1,\lds,S^*_k)|=\el$ initially, then after $\el-(k-1)$
consecutive successful iterations, we have the desired tight \Ttree
$T'$. 

% {
%   \setlength{\intextsep}{0pt}%
  \begin{wrapfigure}{R}{0.35\textwidth}
    % \begin{figure}%{L}{0.5\textwidth}
    \centering
    \includegraphics[width=0.3\textwidth]{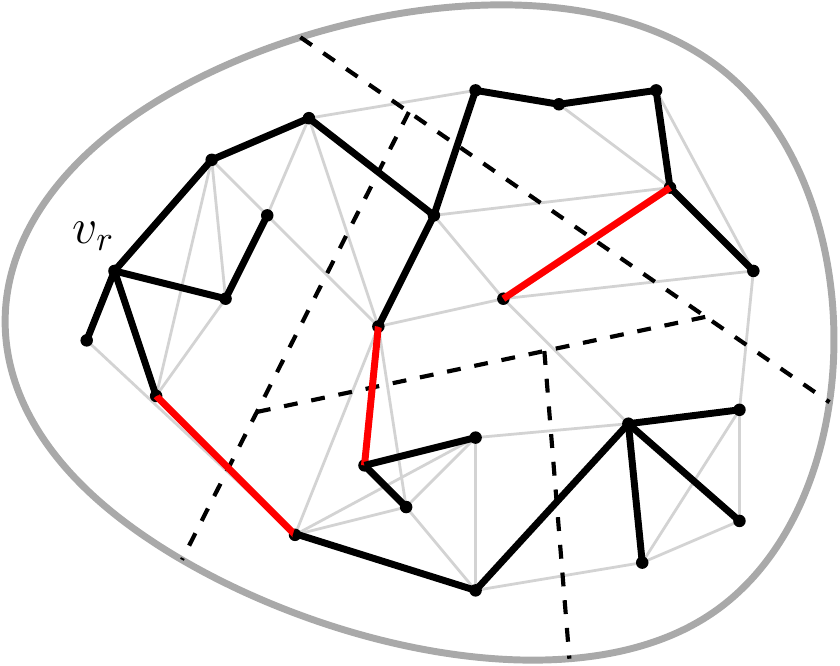}
    \caption{\label{fig:conform} The red edges are deletion-worthy edges
      in this \Ttree; the dashed lines mark the optimal components.}
  \end{wrapfigure}
%}

Assume we know $\el$ beforehand; we can easily discharge this assumption
later. For an intermediate tree $T$ in the algorithm, we say that component
$S^*_i$ is \defn{unsplit} if $S^*_i$ induces exactly one connected
component in $T$, and \defn{split} otherwise. Initially, there are at
most $(k-1)-\el$ split components, possibly fewer if some components
induce many components in $T$. Moreover, if all $\el-(k-1)$ iterations
are successful, all components are unsplit at the end.

\begin{lemma}
  \label{lem:success}
  The probability of any iteration being successful, i.e., reducing the
  number of tree-edges belonging to the optimal cut, is at least
  $\Omega(1/nk^2)$.
\end{lemma}

\begin{proof}
  Each successful iteration has two parts: first we must delete a
  ``deletion-worthy'' edge (which happens with probability $1/(n-1)$),
  and then we add a ``good'' connecting edge (which happens with
  probability $\Omega(1/k^2)$). The former just uses that a tree has
  $n-1$ edges, but the latter must use that there are many good edges
  crossing the resulting cut---a naive analysis may only give
  $\Omega(1/m)$ for the second part.

  % Before we introduce our algorithm of deleting an edge, 
  We first
  describe the edges in $T$ that we would like to delete. These are the
  edges such that if we delete one of them, then we are likely to make a
  successful iteration (after selectively adding an edge back in). We
  call these edges \defn{deletion-worthy}.  Let us first root the tree
  $T=(V,E_T)$ at an arbitrary, fixed root $v_r\in V$. For any edge $e$,
  let $T_e$ denote the subtree below it obtained by deleting the edge
  $e$.

  \begin{definition}
    A \emph{deletion-worthy} edge $e\in E_T$ satisfies the following two
    properties:
    \begin{OneLiners}
    \item[(1)] The edge crosses between two parts of the optimal
      partition, i.e., $e\in E_T(S^*_1,\lds,S^*_k)$.

      \medskip
    \item[(2)] There is exactly one part $S^*_i \in \calS^*$ satisfying
      $S^*_i\cap T_e\ne\emptyset$ and $S^*_i-T_e\ne\emptyset$. In other
      words, exactly one component of $\m S^*$ intersects $T_e$ but is
      not completely contained in $T_e$. Note that, by condition (1),
      $S^*_i$ is necessarily split. 
    \end{OneLiners}
  \end{definition}

  \begin{claim}
    If there is a split component $S_i^*$, there exists a deletion-worthy
    edge $e\in E_T$.
  \end{claim}
  \begin{proof} For each $S^*_i$, contract every connected component of
    $S^*_i$ induced in $T$, so that split components contract to
    multiple vertices. Root the resulting tree at $v_r$, and take a
    vertex $v\in V$ of maximum depth whose corresponding component
    $S^*_i$ is split. It is easy to see that $v\ne v_r$ and the parent
    edge of $v$ in the rooted tree is deletion-worthy.
  \end{proof}
  Finally, we describe the deletion part of our algorithm. The procedure
  is simple: \emph{choose a random edge in $T$ to delete}. With
  probability $\ge1/(n-1)$, we remove a deletion-worthy edge in
  $T$. This gives rise to the $n\inv$ factor in the probability of a
  successful iteration.

  Now we show that, conditioned on deleting a deletion-worthy edge, we
  can selectively add an edge to produce a successful iteration with
  probability $k^{-O(1)}$. In particular, we add a random edge in
  $E_G(T_e,V-T_e)$---i.e., an edge from subtree under $e$ to the rest of
  the vertices---where the probability is weighted by the edge weights
  in $E_G(T_e,V-T_e)$. We show that this makes the iteration successful
  with probability $\Om(1/k^2)$. (Recall that the iteration is
  successful if the number of tree edges lying in the optimal cut
  decreases by~$1$.)

  First of all, it is clear that adding any edge in $E_G(T_e,V-T_e)$
  will get back a tree. Next, to lower bound the probability of success,
  we begin with an auxiliary lemma.

  \begin{claim}
    \label{clm:little-merge}
    Given a set of $k+1$ components $S_1,\lds,S_{k+1}$ that partition $V$,
    we have
    \[w(OPT)\le \bigg(1-\bn{k+1}2\inv\bigg)\cd w(E_G(S_1,\lds,S_{k+1})).\]
  \end{claim}
  \begin{proof}
    Consider merging two components $S_i,S_j$ uniformly at random. Every
    edge in $E(S_1,\lds,S_{k+1})$ has probability $\bn{k+1}2\inv$ of
    disappearing from the cut, so the expected new cut
    is \[\bigg(1-\bn{k+1}2\inv\bigg)\cd E(S_1,\lds,S_{k+1}),\] and
    $w(OPT)$ can only be smaller.
  \end{proof}

  For convenience, define $C:=S^*_i\cap T_e$, where $S_i^*$ is the split
  component corresponding to the deletion-worthy edge $e$ we just
  deleted. Observe that the only edges in $E(T_e,V-T_e)$ that are not in
  $OPT$ must be in $E(C,S^*_i-C)$; this is because, of the components
  $S^*_j$ intersecting $T_e$, only $S^*_i$ is split. Therefore,
  \[w(E(T_e,V-T_e))\le w(OPT)+w(E(C,S^*_i-C)),\]and the probability of
  selecting an edge in $E(C,S^*_i-C)$ is
  \begin{gather}
    \f{w(E(C,S^*_i-C))}{w(E(T_e,V-T_e))}\ge
    \f{w(E(C,S^*_i-C))}{w(OPT)+w(E(C,S^*_i-C))}.\label{eq:1}
  \end{gather}
  \begin{claim}
    $w(E(C,S^*_i-C)) \ge \lp \lp1-\bn{k+1}2\inv\rp\inv-1 \rp \cd w(OPT) =
    \Om(1/k^2)\cd w(OPT)$.
  \end{claim}
  \begin{proof}
    The set of edges $OPT \cup E(C,S^*_i-C)$ cuts the graph $G$ into
    $k+1$ components. Claim~\ref{clm:little-merge} implies this set has
    total weight $\geq \lp1-\bn{k+1}2\inv\rp\inv w(OPT)$. Observing that
    the edges of $OPT$ and $E(C,S^*_i-C)$ are disjoint from each other
    completes the proof.
    % allows
    % us to say
    % $w(E(C,S^*_i-C)) \ge \lp1-\bn{k+1}2\inv\rp\inv-1)\cd w(OPT)$,
    % proving
    % the lemma.
  \end{proof}

  Using the above claim in~(\ref{eq:1}) means the probability of
  selecting an edge in $E(C,S^*_i-C)$ is $\Om(1/k^2)$. Hence the
  probability of an iteration being successful is $\Omega(1/(nk^2))$,
  completing the proof of Lemma~\ref{lem:success}.
\end{proof}

Since we have $\ell$ iterations, the probability that each of them is
successful is $\el^{-O(\el)}n^{-\el}$. If we repeat this algorithm
$\el^{O(\el)}n^\el\log n$ times, then with probability $1 - 1/\poly(n)$,
one of the final trees $T'$ will satisfy
$|E_{T'}(S^*_1,\lds,S^*_k)|=k-1$. We can remove the assumption of
knowing $\el$ by trying all possible values of $\el\in[k-1,2k-2]$,
giving a collection of $k^{O(k)}n^{k-1}\log n$ trees in running time
$k^{O(k)}n^{k-1}m\log n$. This completes the proof of
Lemma~\ref{lem:Merge}.

\subsection{Solving \kcut on ``Tight'' Trees}\label{sec:Remove}

In the previous section, we found a collection of $\approx n^k$ trees
such that, with high probability, the intersection of one of these trees
with the optimal $k$-cut $OPT$ consists of only $k-1$ edges. In this
section, we show that given this tree we can find the optimal $k$-cut in
time $\approx n^{\omega k/3}$. This will follow from
Lemma~\ref{lem:Remove} below. In this section, we restrict the edge
weights of our graph $G$ to be positive integers in $[W]$.

\begin{restatable}{lemma}{Remove}\label{lem:Remove}
  There is an algorithm that takes a tree $T$ and outputs, from among
  all partitions $\{S_1,\lds,S_k\}$ that satisfy
  $|E_T(S_1,\lds,S_k)|=k-1$, a partition
  $\m S^\dag:=\{S^\dag_1,\lds,S^\dag_k\}$ minimizing  the number of
  inter-cluster edges $E_G(S^\dag_1,\lds,S^\dag_k)$, 
  in time $\tO(k^{O(k)}n^{\lf(k-2)/3\rf\om+2+(k-2)\bmod3}W)$.
\end{restatable}

%\Remove*

Given a tree $T = (V, E_T)$ and a set $F\s E_{T}$ of tree edges,
deleting these edges gives us a vertex partition
$\calS_F = \{S_1, \ldots, S_{|F|+1}\}$. Let $\cut(F)$ be the set of
edges in $G$ that go between the clusters in $\calS_F$; i.e., 
\begin{gather}
  \cut(F) := E(S_1, \ldots, S_{|F|+1}).\label{eq:cut-def}
\end{gather}
Put another way, these are the edges
$(u,v) \in E$ such that the unique $u$-$v$ path in $T$ contains an edge
in $F$. Note that Lemma~\ref{lem:Remove} seeks a set $E^\dag \s E_T$ of size
$k-1$ that minimizes $w(\cut(F))$.

\subsubsection{A Simple Case: Incomparable Edges}
\label{sec:incomp-rand}

Our algorithm builds upon the algorithm of Ne\v{s}et\v{r}il and
Poljak~\cite{nevsetvril1985complexity} for \kclique, using Boolean
matrix multiplication to obtain the speedup from the naive $O(n^k)$
brute force algorithm. It is instructive to first consider a restricted
setting to highlight the similarity between the two algorithms. This
setting is as follows: we are given a vertex $v_r\in V$ and the promise
that if the input tree $T=(V,E_T)$ is rooted at $v_r$, then the optimal
$k-1$ edges $E^\dag:=E_T(S^\dag_1,\lds,S^\dag_k)$ to delete are
\textit{incomparable}. By incomparable, we mean any root-leaf path in
$T$ contains at most one edge in $E^\dag$.

Like the algorithm of~\cite{nevsetvril1985complexity}, our algorithm
creates an auxiliary graph $H=(V_H,E_H)$ on $O(n^{\lc k/3\rc})$
nodes. Our graph construction differs slightly in that it always
produces a tripartite graph, and that this graph has edge weights. In
this auxiliary graph, we will call the vertices \emph{nodes} in order to
differentiate them from the vertices of the tree. 
\begin{itemize}
\item The nodes in graph $H$ will form a tripartition
  $V_1\cup V_2\cup V_3=V_H$. For each $r$, let $\m F_r\s 2^E$ be the
  family of all sets of exactly $r$ edges in $E_T$ that are pairwise
  incomparable in $T$. For each $i=1,2,3$, define
  $r_i:=\lf\frac{(k-1)+(i-1)}3\rf$ so that $r_1+r_2+r_3=k-1$. For each
  $i=1,2,3$ and each $F\in\m F_{r_i}$, add a node $v_i^F$ to $V_i$
  representing set $F$. 
\item Consider a pair $(V_a,V_b)$ of parts in the tripartition with
  % $a\neq b,$ $a,b \in [3]$.
  $(a,b) \in (1,2), (2,3), (3,1)$. 
  Consider a pair of sets
  $F^a:=\{e^a_1,\lds,e^a_{r_a}\}\in\m F_{r_a}$,
  $F^b:=\{e^b_1,\lds,e^b_{r_b}\}\in\m F_{r_b}$; recall these are sets of
  $r_a$ and $r_b$ incomparable edges in $T$.  If the edges in $F^a$ are
  also pairwise incomparable with the edges in $F^b$, then add an edge
  $(v_a^{F^a},v_b^{F^b})\in V_a\times V_b$ of weight
\[w_H(v^{F^a}_a,v^{F^b}_b):=\sum_{i=1}^{r_a}w(E(T_{e^a_i},V-T_{e^a_i}))
  - \sum_{i=1}^{r_a}\sum_{j=i+1}^{r_a}w(E(T_{e_i^a},T_{e_j^a})) -
  \sum_{i=1}^{r_a}\sum_{j=1}^{r_b}w(E(T_{e^a_i},T_{e^b_j})).\]
\end{itemize}

Observe that every triple of nodes in graph $H$ that form a triangle
together represent $r_1+r_2+r_3=k-1$ many incomparable edges. Moreover,
the weights are set up so that for any triangle
$(v^{F^1}_1,v^{F^2}_2,v^{F^3}_3)\in V_1\times V_2\times V_3$ such that
$F:=F^1\cup F^2\cup F^3=\{e_1,\lds,e_{k-1}\}$, the total weight of the
edges is equal
to
  \begin{gather}
  w_H(v^{F^1}_1,v^{F^2}_2)+w_H(v^{F^2}_2,v^{F^3}_3)+w_H(v^{F^3}_3,v^{F^1}_1)=\sum_{i=1}^{k-1}
  w(E(T_{e_i},V-T_{e_i}))-\sum_{i=1}^{k-1}\sum_{j={i+1}}^{k-1}w(E(T_{e_i},T_{e_j})).
  \label{eq:TriangleWeight}\end{gather}
A straightforward counting argument shows that this is exactly
$w(E(T_{e_1},\lds,T_{e_{k-1}}))=\cut(F)$, the solution value of cutting
the edges in $F$.

Hence, the problem reduces to computing a minimum weight triangle in
graph $H$. While the minimum weight triangle problem is unlikely to
admit an $O(N^{3-\e})$ time algorithm on a graph with $N$ vertices with
arbitrary edge weights, the problem does admit an $\tO(MN^{\om})$ time
algorithm when the graph has integral edge weights in the range
$[-M,M]$~\cite{williams2010subcubic}. Since the original graph $G$ has
integral edge weights in $[W]$, the edge weights in $H$ must be in the
range $[-O(Wm),O(Wm)]$. Therefore, we can set $N:=O(n^{\lc (k-1)/3\rc})$
and $M:=O(Wm)$ to obtain an $\tO(Wn^{\lc (k-1)/3\rc\om}m)$ time
algorithm in this restricted setting.

\subsubsection{The General Algorithm}
\label{sec:general-rand}

\newcommand{\st}{{\small\textsf{State}}}

Now we prove Lemma~\ref{lem:Remove} in full generality, and show how to
find $E^\dag$. The ideas we use here will combine the
matrix-multiplication idea from the restricted case of incomparable
edges, together with dynamic programming.
\begin{quote}
  Given a tree edge $e\in E_T$, and an integer $s\in[k-2]$, let
  $\st(e,s)$ denote a set of edges~$F$ in subtree $T_e$ such that
  $|F|=s-1$ and $\cut(\{e\}\cup F)$ is minimized.
\end{quote}
In other words, $\st(e,s)$ represents the optimal way to cut edge $e$
along with $s-1$ edges in $T_e$. For ease of presentation, we assume
that this value is unique. Observe that, once all of these states are computed, the remaining problem boils down to choosing an integer $\el\in[k-1]$, integers $s_1,\lds,s_\el$ whose sum is $k-1$, and incomparable edges $e_1,\lds,e_\el$ that minimizes
\[ \cut \lp \bigcup_{i=1}^\el \st(e_i,s_i) \rp= \sum_{i=1}^{k-1}\st(e_i,s_i) - \sum_{i=1}^{k-1}\sum_{j=i+1}^{k-1}w(E(T_{e_i},T_{e_j})) .\]
Comparing this expression to (\ref{eq:TriangleWeight}) suggests that this problem is similar to the incomparable case in \S\ref{sec:incomp-rand}, a connection to be made precise later.

  We now compute states for all edges
$e\in E_T$, which we do from bottom to top (leaf to root).  When $e$ is
a leaf edge, the states are straightforward: $\st(e,1)=\cut(\{e\})$ and
$\st(e,s)=\infty$ for $s>1$. Also, for each edge $e\in E_T$, define
$\desc(e)$ to be all ``descendant edges'' of $e$, formally defined as
all edges $f\in E_T-e$ whose path to the root contains edge $e$.

Fix an edge $e\in E_T$ and an $s\in[k-2]$, for which we want to compute
$\st(e,s)$. Suppose we order the edges in $T_e$ in an arbitrary but
fixed order. Let us now figure out some properties for this (unknown)
value of $\st(e,s)$.  As a thought experiment, let $F^\dag$ be the list
of all the ``maximal'' edges in $\st(e,s)$---in other words,
$f\in F^\dag$ iff $f \in \st(e,s)$ and $f\notin \desc(f')$ for all
$f'\in\st(e,s)$. Let $\el^\dag:=|F^\dag|$ and
$F^\dag=(e^\dag_1, \lds, e^\dag_{\el^\dag})$ be the sequence in the
defined order, and for each $e^\dag_i$, let
$s^\dag_i:=1+|\desc(e^\dag_i) \cap \st(e,s) |$. Observe that
$\sum_is^\dag_i=s-1$, and that we must satisfy
\begin{gather}
  \st(e,s)=\bigcup_{i=1}^{\el^\dag}\lp\{e^\dag_i\}\cup\st(e^\dag_i,s^\dag_i)\rp.
\end{gather}
Also,
\begin{align*}
w(\st(e,s))= E(T_e,V-T_e) &+ \sum_{i=1}^{\el^\dag}
  w(E(G[T_e]) \cap \cut\lp\{e^\dag_i\}\cup\st(e^\dag_i,s^\dag_i)\rp) \\
&-
  \sum_{i=1}^{\el^\dag}
  \sum_{j=i+1}^{\el^\dag}w(E_{G[T_e]}[T_{e^\dag_i},T_{e^\dag_j}]), \end{align*}
since the only edges double-counted in the first summation of
$w(\st(e,s))$ are those connecting different
$T_{e^\dag_i},T_{e^\dag_j}$. 

Given these ``ideal'' values $\el^\dag$ and $\{s_i^\dag\}$, our
algorithm repeats the following procedure multiple times:
\begin{OneLiners}%[$\centerdot$]
 
\item Pick a number $\el$ uniformly at random in $[s-1]$.  Then, let
  function $\sigma:[\el]\to[s-1]$ be chosen uniformly at random among all
  $\le(s-1)^\el$ possible functions satisfying
  $\sum_{i=1}^\el \sigma(i)=s-1$. With probability
  $\ge(s-1)^{-(\el^\dag+1)}=k^{-O(k)}$, we correctly guess
        $\el=\el^\dag$ and $\sigma(i)=s^\dag_i$ for each $i\in[\el]$.\footnote{Of course, we could instead brute force over all $k^{O(k)}$ possible choices of $\el$ and $\sigma$.}
\item Construct an auxiliary graph $H$ as follows. As in
  \S\ref{sec:incomp-rand}, $H$ has a tripartition
  $V_1\cup V_2\cup V_3=V_H$, and assume there is an arbitrary but fixed
  total ordering on the edges of the tree. For each $r$, let
  $\m F_r\s 2^E$ be the family of all sets of exactly $r$ edges in $E_T$
  that are pairwise incomparable in $T$. For each $i=1,2,3$, let
  $r_i:=\lf\frac{\el+(i-1)}3\rf$ so that $r_1+r_2+r_3=\el$, and for each
  $F\in\m F_{r_i}$, add a node $v_i^F$ to $V_i$ representing the edges
  $F$ \textit{as a sequence in the total order}.

  Also, define $R_i:=\sum_{j=1}^{i-1}r_i$ for $i=1,2,3,4$. Note that
  $R_1=0$ and $R_4 = r_1 +r_2+r_3 = \ell$.  Our intention is map the integer values
  $\{\sigma(R_i+1), \sigma(R_i+2), \ldots, \sigma(R_{i+1})\}$ to the
  sequences represented by nodes in $V_i$, as we will see
  later. Consider each tripartition pair $(V_a,V_b)$ with $(a,b) \in
  (1,2), (2,3), (3,1)$. For each pair $F^a\in\m F_{r_a}$, $F^b\in\m F_{r_b}$ represented
  as ordered sequences $F^a = (e^a_1,\lds,e^a_{r_a})$ and
  $F^b = (e^b_1,\lds,e^b_{r_b})$, if the edges in $F^a$ are pairwise
  incomparable with the edges in $F^b$, then add an edge
  $(v_a^{F^a},v_b^{F^b})\in V_a\times V_b$ in the auxiliary graph of
  weight
  \begin{multline}
    w_H(v^{F^a}_a,v^{F^b}_b):=\sum_{i=1}^{r_a}w\Big(\st\big(e^a_i,{\sigma}(R_a+i)\big)\Big) \\
    -
    \sum_{i=1}^{r_a}\sum_{j=i+1}^{r_a}w(E_{G[T_e]}(T_{e_i^a},T_{e_j^{a}}))
    -
    \sum_{i=1}^{r_a}\sum_{j=1}^{r_b}w(E_{G[T_e]}(T_{e^a_i},T_{e^b_j})).
  \end{multline}
  For any triangle
  $(v^{F^1}_1,v^{F^2}_2,v^{F^3}_3)\in V_1\times V_2\times V_3$ such that
  $F:=F^1\cup F^2\cup F^3$ has ordered sequence $(e_1,\lds,e_\el)$, the
  total weight of the edges is equal to
  \begin{multline}
    w_H(v^{F^1}_1,v^{F^2}_2)+w_H(v^{F^2}_2,v^{F^3}_3)+w_H(v^{F^3}_3,v^{F^1}_1)
     \\=\sum_{i=1}^\el
    w(\st(e^a_i,{\sigma}(i)))-\sum_{i=1}^\el\sum_{j={i+1}}^\el
    w(E_{G[T_e]}(T_{e_i},T_{e_j})).
  \end{multline}
\end{OneLiners}

A straightforward counting argument shows that this is
exactly
\[w\Big(\cut\big(\{e\}\cup \bigcup_{i=1}^\el\st(e_i,\sigma(i))\big)\Big)
- w(E(T_e,V-T_e)).\] Thus, the weight of each triangle, with
$w(E(T_e,V-T_e))$ added to it, corresponds to the cut value of one
possible solution to $\st(e,s)$. Moreover, if we guess $\ell$ and
$\sigma: [\ell] \to [s-1]$ correctly, then this triangle will exist in
auxiliary graph $H$, and we will compute the correct state if we compute
the minimum weight triangle in $\tO(Wn^{\lc \el/3\rc\om}m)$ time.  Since
the probability of guessing $\ell, \sigma(\cdot)$ correctly is
$k^{-O(k)}$, we repeat the guessing $k^{O(k)}\log n$ times to succeed
w.h.p.\ in time $\tO(k^{O(k)}n^{\lc (k-2)/3\rc\om}mW)$. This concludes the computation of each $\st(e,s)$; since there are $O(kn)$ such states, the total running time becomes $\tO(k^{O(k)}n^{\lc(k-2)/3\rc\om+1}mW)$. 

Lastly, to compute the final \kcut value, we let $s:=k-1$ and construct the same auxiliary graph $H$, except that $k-2$ is replaced by $k-1$ and the relevant graph $G[T_e]$ becomes the entire $G$. By the same counting arguments, the weight of triangle
  $(v^{F^1}_1,v^{F^2}_2,v^{F^3}_3)\in V_1\times V_2\times V_3$ such that
  $F:=F^1\cup F^2\cup F^3$ has ordered sequence $(e_1,\lds,e_\el)$ is exactly
\[ w\Big(\cut\big(\{e\}\cup \bigcup_{i=1}^\el\st(e_i,\sigma(i))\big)\Big)
 .\]
Again, by repeating the procedure $k^{O(k)}\log n$, we compute an optimal \kcut w.h.p., in time $\tO(k^{O(k)}n^{\lc (k-1)/3\rc\om}mW)$. Note that this time is dominated by the running time $\tO(k^{O(k)}n^{\lc (k-2)/3\rc\om+1}mW)$ of computing the states.

In order to get the runtime claimed in Theorem~\ref{thm:FasterKCut}, we
need a couple more ideas---however, they can be skipped on the first
reading, and we defer them to the Appendix~\ref{sec:app-random}.

%%% Local Variables:
%%% mode: latex
%%% TeX-master: "main"
%%% End:

\section{A Faster Deterministic Algorithm}
\label{sec:det-algo}

\newcommand{\map}[1]{\sigma(#1)}

In this section, we show how to build on the randomized algorithm of the
previous section and improve it in two ways: we give a deterministic
algorithm, with a better asymptotic runtime. (The algorithm of the
previous section has a better runtime for smaller values of $k$.)
Formally, the main theorem of this section is the following:

\EvenFasterKCut*

Our main idea is a more direct application of matrix multiplication,
without paying the $O(n^k)$ overhead in the previous section.
Instead of converting a given T-tree to a ``tight'' tree where
matrix multiplication can be combined with dynamic programming,
with only $n^{O(\log k)}$ overhead, 
we partition the given T-tree to subforests that are amenable to
direct matrix multiplication approach.

As in \S\ref{sec:randomized} we build on the framework of
Thorup~\cite{Thorup08}, where the \kcut problem reduces to $n^{O(1)}$
instances of the following problem: given the graph $G$ and a spanning
tree $T$, find a way to cut $\le2k-2$ edges from $T$, and then merging
the connected components of $T$ into $k$ connected components, that
minimizes the number of cut edges in $G$. Again, the optimal $k$-cut is
denoted by $\calS^* = \{S_1^*, \ldots, S_k^*\}$. 

For the rest of this section, let $T$ be some spanning tree in the
instance that crosses the optimal $k$-cut in $(r-1) \leq 2k-2$ edges.
If we delete these $r-1$ edges from $T$, this gives us $r$ components,
which we denote by $C^*_1,\lds,C^*_r$ --- these are a refinement of
$\calS^*$, and hence can be then be merged together to give us
$\calS^*$. Let $E_T^*:=E_T(C_1^*,\lds,C^*_r) = E_T(S_1^*,\lds,S^*_k)$ be
these $r-1$ cut edges in $T$. 

\subsection{Balanced Separators}

We first show the existence of a small-size \textit{balanced separator}
in the following sense: there exist forests $F_1,F_2,F_3$ whose vertices
partition $V(T)$, such that
\begin{OneLiners}
  \item[(i)] we can delete $O(\log k)$ edges in $T$ to get the forests,
  i.e., $|E(T)-\bigcup_{i=1}^3E(F_i)|=O(\log k)$, and 
\item[(ii)] we want to cut few edges from each forest, i.e.,
  $|E(F_i)\cap E_T^*|\le\lc 2k/3\rc$ for each $i$.
\end{OneLiners}
Of course, small-size balanced edge separators typically do not exist in
general trees, such as if the tree is a star. So we first apply a
degree-reducing step. This operation reduces the maximum degree of the
tree to $3$, at a cost of introducing ``Steiner'' vertices, which are
handled later.

\begin{lemma}[Degree-Reduction]
  \label{lem:DegreeReduce}
  Given a tree $T = (V_T,E_T)$, we can construct a tree
  $T' = (V_{T'}, E_{T'})$, where $V_{T'} = V_T \cup X$, where $X$ are called the
  Steiner vertices, such that
  \begin{OneLiners}
  \item[1.] $T'$ has maximum degree $3$.
  \item[2.] $|V(T')|\le 2|V(T)|$
  \item[3.] For every way to cut $r$ edges in $T$ and obtain components
    $C_1,\lds,C_{r+1}$, there is a way to cut $r$ edges in $T'$ and
    obtain components $C'_1,\lds,C'_{r+1}$ such that each $C_i$ is
    precisely $C_i' \cap V_T$. 
  \end{OneLiners}
\end{lemma}

\begin{proof}
  Root the tree $T$ at an arbitrary root, and select any non-Steiner
  vertex $v\in V_T$ with more than two children. Replace the star
  composed of $v$ and its children with an arbitrary binary tree with
  $v$ as the root and its children as the leaves. This process does not
  introduce any new vertex with more than two children, so we can repeat
  it until it terminates, giving us a tree $T'$ of maximum degree
  $3$. Every star of $z$ edges adds exactly $z-1$ Steiner nodes, and
  there are $\le|V_T|-1$ edges initially, so $\le |V_T|-2$ Steiner
  vertices are added throughout the process, and $|V_{T'}|\le
  2|V_{T}|$. Finally, if we cut some $r$ edges $(u_i,v_i)\in E_T$ where
  $v_i$ is the parent of $u_i$, then we can cut the parent edge of each
  $u_i$ in $T'$ to obtain the required components.
\end{proof}

% Hence, our algorithm first applies this degree-reducing procedure,
% working with the tree $T'$ of degree $3$.
Having applied Lemma~\ref{lem:DegreeReduce} to $T$ to get $T'$,
Property~(3) shows that we can still delete $\le2k-2$ edges in $T'$ to
obtain the components of the optimal solution before merging. To avoid
excess notation, we assume that $T$ itself is a tree of degree $\le3$,
possibly with Steiner nodes. From now on, our task is to delete
$\le2k-2$ edges of $T$ and merge them into $k$ components, \textit{each
  of which containing at least one non-Steiner vertex}, that minimizes
the number of cut edges in $G$. To show that the aforementioned forests
$F_1,F_2,F_3$ exist in the new tree $T$, we introduce the following
easy lemma:

\begin{lemma}\label{lem:BalancedSeparator}
  Let $T$ be a tree of degree $\le 3$ and $F\s E(T)$ be a subset of the
  edges. For any integer $r \in [1,|F|-1]$, there exists a vertex
  partition $A,B$ of $V(T)$ such that $|E_T(A,B)|=O(\log (r+1))$, and the
  induced subgraphs $T[A]$ and $T[B]$ have at most $r$ and at most $|F|-r$ edges
  from $F$, respectively.
\end{lemma}
\begin{proof}
  We provide an algorithm that outputs a collection of $O(\log r)$
  disjoint subtrees whose union comprises $A$. Root $T$ at a degree-$1$
  vertex, and find a vertex of maximal depth whose rooted subtree
  contains $>r$ edges from $F$. The degree condition ensures that $v$
  has $\le2$ children, and by maximality, all of $v$'s children have
  $\le r$ edges in $F$ in their subtrees. Moreover, the edges in $T_v$
  is precisely the union of the edge sets $E(T_u)\cup\{(u,v)\}$ for all
  children $u$ of $v$. For convenience, define
  $E^+(T_u):=E(T_u)\cup\{(u,v)\}$ for a child $u$ of $v$. So there must
  be a child $u$ of $v$ satisfying $|E^+(T_u) \cap F| \in (r/2,r]$.

  If $|E^+(T_u) \cap F|=r$, then $(A,B)=(V(T_u),V(T)-V(T_u))$ is a
  satisfying partition with $|E_T(A,B)|=1$, and we are done. Otherwise,
  recurse on the tree $T':=T[V(T)-V(T_u)]$ where we remove $(u,v)$ and
  the subtree below it, with the parameters $r':=r-|E^+(T_u)\cap F|$ and
  $F':= F \setminus E^+(T_u)$ to get partition $(A',B')$, and set
  $A:=A'\cup V(T_u)$ and $B:=B'$. By recursion, we guarantee that
  \begin{align*}
    |E(T[A])\cap F|&\le|E(T[A'])\cap F'|+|E^+(T_u)\cap F|\\
                   &\le (r-|E^+(T_u)\cap F|) + |E^+(T_u)\cap F| =r
  \end{align*}
  and
  \begin{align*}
    |E(T[B])\cap F'|=|E(T[B'])\cap F'|&\le |F'|-(r-|E^+(T_u)\cap F|)\\
                                      &= (|F|-|E^+(T_u)\cap F|)-(r-|E^+(T_u)\cap F|)
                                        =|F|-r.
  \end{align*}
  Since the value of $r$ drops by at least half each time, there are
  $O(\log r)$ steps of the recursion. Each step can only add the
  additional edge $(u,v)$ to $|E_T(A,B)|$, so $|E_T(A,B)|=O(\log r)$.
\end{proof}

\begin{corollary}\label{cor:BalancedTripartition}
  There exist forests $F_1,F_2,F_3$ whose vertices partition $V(T)$ such
  that
  \begin{OneLiners}
  \item[(i)] the number of crossing edges is
    $|E(T)-\bigcup_{i=1}^3E(F_i)|=O(\log |E_T^*|)$, and
  \item[(ii)] $|E(F_i)\cap E_T^*|\le\lc |E_T^*|/3\rc$ for each $i$.
  \end{OneLiners}

\end{corollary}
\begin{proof}
  We apply Lemma~\ref{lem:BalancedSeparator} with $F:=E_T^*$ and
  $r:=\lc|E_T^*|/3\rc$ to obtain the separation $(A,B)$, and then set
  $F_1:=T[A]$. Before applying the lemma again on $B$, we first connect
  the connected components of $B$ arbitrarily into a tree; let $F^+$
  denote the added edges. Then, we apply with $F:=E_T^*-E[F_1]$ and
  $r:=\lc|E_T^*|/3\rc$ to obtain separation $(A',B')$, and then set
  $F_2:=T[A']-F^+$ and $F_3:=T[B']-F^+$.
\end{proof}

Given this result, our algorithm starts by trying all possible
$n^{O(\log k)}$ ways to delete $O(\log k)$ edges of $T$ and partition
the connected components into three forests. By
Corollary~\ref{cor:BalancedTripartition}, one of these attempts produces
the desired $F_1,F_2,F_3$ satisfying the two properties.

\subsection{Matrix Multiplication}
\label{sec:mat-mult}

The balanced partitioning procedure from the previous section gives us
three forests $F_1,F_2,F_3$, such that the optimal solution cuts at most
$2k/3$ edges in each --- and then combines the resulting pieces
together. The algorithm now computes these solutions separately for each
forest, and then uses matrix multiplication to combine these solutions
together.

Indeed, for each $F_i\in\{F_1,F_2,F_3\}$, the algorithm computes all
$O(n^{\lc 2k/3\rc})$ ways to cut $\le\lc 2k/3\rc$ edges in $F_i$,
followed by all $3^{O(k)}$ ways to label each of the
$\le\lc 2k/3\rc+O(\log k)$ connected components with a label in
$[k]$. For each one forest, note that some of these components might
only contain Steiner vertices of the tree; we call these the
\emph{Steiner components}, and the other the \emph{normal
  components}. For each subset $S\s[k]$, let $\m F_i^S$ denote all
possible ways to cut and label $F_i$ in the aforementioned manner such
that the set of labels that are attributed to at least one normal
component is precisely~$S$.

The algorithm now enumerates over every possible triple of subsets
$S_1,S_2,S_3\s [k]$ (not necessarily disjoint) whose union is exactly
$[k]$.  Note that there are at most $7^k$ of these triples. For each
triple $S_1, S_2, S_3$, we construct the following tripartite auxiliary
graph $H=(V_H,E_H)$ on $O(k^{O(k)}n^{\lc 2k/3\rc})$ vertices, 

with tripartition
$V_H = V_1 \uplus V_2 \uplus V_3$. For each $i=1,2,3$, each element in
$\m F_i^{S_i}$ is a tuple $(X_i, \sigma_i)$ where $X_i \sse F_i$ is a
set of edges that we cut from $F_i$, and $\sigma_i$ is a labeling of the
normal components in the resulting forest so that the label set is
exactly $S_i$. Now for each $(X,\sigma) \in\m F_i^{S_i}$, add a node
$v_i^{(X, \sigma)}$ to $V_i$. Moreover, for each tripartition pair
$(V_a,V_b)$ with $(a,b) \in (1,2),(2,3),(3,1)$, and for each way
$(X_a, \sigma_a) \in\m F_a^{S_a}$ to cut $F_a$ into components
$C^a_1,\lds,C^a_{r_a}$ with labels $\sigma_a(1),\lds,\sigma_a(r_a)$, and
for each way $(X_b, \sigma_b) \in\m F_b^{S_b}$ to cut $F_b$ into
components $C^b_1,\lds,C^b_{r_b}$ with labels
$\sigma_b(1),\lds,\sigma_b(r_b)$, we add an edge
$(v_a^{(X_a, \sigma_a)},v_b^{(X_b, \sigma_b)})\in V_a\times V_b$ of weight
\begin{multline}
  w_H(v_a^{(X_a, \sigma_a)},v_b^{(X_b, \sigma_b)}):=
  \sum_{i=1}^{r_a}\sum_{j=i+1}^{r_a} \1[\sigma_a(i)\ne \sigma_a(j)]\cd
  w(E_G[C_i^a,C_j^a]) \\ +
  \sum_{i=1}^{r_a}\sum_{j=1}^{r_b}\1[\sigma_a(i)\ne \sigma_b(j)]\cd
  w(E_G[C_i^a,C_j^b]),
\end{multline}
where $\1$ is the indicator function, taking value $1$ if the
corresponding statement is true and $0$ otherwise. Finally, the
algorithm computes the minimum weight triangle in $H$.

A straightforward counting argument shows that the weight of each
triangle
$(v_1^{(X_1, \sigma_i)},v_2^{(X_2, \sigma_2)},v_3^{(X_3, \sigma_3)})$ in
$H$ is exactly the value of the cut in $G$ obtained by merging all
components in $F^1,F^2,F^3$ with the same label together. In particular,
for the correct triple $S_1,S_2,S_3$ for $E_T^*$, there is a triangle in
$H$ whose weight is the cost of the optimal solution, and the algorithm
will find it, proving the correctness of the algorithm.

As for running time, the algorithm has an $n^{O(\log k)}7^{k}$ overhead
for the guesswork of finding the forests $(F_1, F_2, F_3)$ and the
correct triple $(S_1, S_2, S_3)$ of subsets of labels. This is followed
by computing matrix multiplication on a graph with
$k^{O(k)}n^{\lc2k/3\rc}$ nodes, with edge weights in
$[-Wm,Wm]$. Altogether, this takes $k^{O(k)}n^{(2\om/3+\e)k+O(1)}W$ for
any $\e>0$, proving Theorem~\ref{thm:EvenFasterKCut}.

%%% Local Variables:
%%% mode: latex
%%% TeX-master: "main"
%%% End:

\section{An $(1+\e)$-Approximation Algorithm}
\label{sec:approx}

We now give a $(1+\e)$-approximation algorithm for the $k$-cut problem
that achieves a running time better than both the previous
algorithms. Moreover, the ideas we develop here allow us to get a better
constant-factor approximation for \kcut in FPT time. The main theorem we
prove is the following:

\ApproxKCut*

The runtime of $n^{k}$ comes from the reduction given in
Lemma~\ref{lem:Merge} that, given a \Ttree---i.e., a tree that crosses
the optimal $k$-cut in $\leq 2k-2$ edges---alters it to return a
collection of $n^{k+O(1)}$ trees that contain at least one tight \Ttree
$T$, i.e., one that crosses the optimal $k$-cut in exactly $k-1$
edges. How do we find the right $k-1$ edges to cut, to minimize the
total weight of edges in $G$ that go between different components?  It
is this problem that we give an FPT-PTAS for: we show how to approximate
\kcut on tight {\Ttree}s to within a $(1+\e)$-factor in time FPT in $k$:

\begin{lemma}\label{lem:fpt-approx-tight}
Given a tree $T$ satisfying $|E_T(S^*_1,\lds,S^*_k)|=k-1$, there is a deterministic $(1+\e)$-approximation algorithm for the \kcut problem with runtime $(k/\e)^{O(k)}\poly(n)$.
\end{lemma}

% Like \S\ref{sec:Remove}, this section considers the problem of cutting
% $k-1$ edges of a tree $T$ . Recall that, by Lemma~\ref{lem:Merge}, we
% can reduce the \kcut problem to $n^{k+O(1)}$ instances of this
% problem. In this section, we solve this problem $(1+\e)$-approximately
% in time FPT in $k$.
%
% given an edge $e$ in the tree $T$, we define its \defn{cut value} to
% be the size of the cut in $G$ of the two connected components of
% $T-e$.

In this section, we only aim at a running time of $2^{\poly(k/\e)}\poly(n)$, in an effort to display our main ideas in a more streamlined fashion. To prove the running time required for Lemma~\ref{lem:fpt-approx-tight}, we defer the additional ideas to \S\ref{sec:s4-improvements}.

Firstly, we need an estimate for $w(OPT)$, for which a coarse
approximation algorithm suffices. Indeed, let $M$ be the value of a
$2$-approximation algorithm to \kcut~\cite{SV95}, so that our algorithm
knows $M$ and $w(OPT)\le M\le 2w(OPT)$. Also, recall
from~(\ref{eq:cut-def}) that given tree $T$ and a set of edges
$F \sse E(T)$, if $S_1, S_2, \ldots, S_{|F|+1}$ is the vertex partition
obtained by deleting edges $F$ from tree $T$, then
$\cut(F) = E_G(S_1, \ldots, S_{|F|+1})$ denotes the edges in the
underlying graph $G$ that cross this partition. We make the following
simple observation.

\begin{observation}
  \label{obs:M}
  For each of the $k-1$ edges $e\in E_T[S_1^*,\lds,S_k^*]$, $\cut(\{e\}) \le M$.
\end{observation}

This allows us to contract all edges $e\in E(T)$ with $\cut(\{e\})>M$,
since they cannot be cut in the optimal solution. Henceforth, assume
that every edge $e\in E(T)$ has $\cut(\{e\})\le M$.

\subsection{The Game Plan}
\label{sec:game-plan}

%\paragraph{Dynamic programming.}
We want to apply dynamic programming on the tree $T$, which we root at
an arbitrary vertex. The first question to ask is: for each subtree
$T_v$, $v\in V(T)$, what dynamic programming states should we compute
and store? As is typical in dynamic programming algorithms, we want our
states to be as informative as possible, so that computing new states
can be done efficiently. However, we also want a small number of states.
Hence, we need to find a balance between a sparse representation of
states and a fast way to compute them.

For each vertex $v\in V(T)$ and integer $s\in[k-1]$, we want to store a
collection of states for $v$ such that one of them provides information
about $E_T[S_1^*,\lds,S_k^*]$ when restricted to the scope of $T_v$. One
way is the following: for $v, s, \{e_1, \ldots, e_\ell\}$, find the best
way to cut $s$ edges in the subtree below $v$, given that the cut edges
closest to $v$ are these $\ell$ incomparable edges. (We formalize this
below in \S\ref{sec:DQ}.) This dynamic program captures the problem
exactly. But since $\ell$ could be close to $k$ (for star-like graphs),
there could be roughly $n^{k-1}$ states, which would be no better than
brute-force search. Indeed, the reduction from clique shows we do not
expect to solve the problem exactly on stars faster than $n^{\om k/3}$
time; see \S\ref{sec:lower-bounds}. Hence, we compress the number of
states at a loss of a $(1+\e)$-approximate factor. Indeed, we represent
each ``true'' state $(u_1,\lds,u_\el)$ approximately with a ``small''
family of \textit{representative} states---i.e., a family with size that
is FPT in $k$.

\subsection{The Ideal Dynamic Program}
\label{sec:DQ}

\newcommand{\DQ}{\textsf{\small ExactDP}\xspace}
\newcommand{\DPStar}{\textsf{\small SmallDP}\xspace}
\newcommand{\DPee}{\textsf{\small PolyDP}\xspace}
\newcommand{\round}{\textsf{\small round}\xspace}

We extend the definition of $\cut$ from a mapping for edge sets given
in~(\ref{eq:cut-def}) to vertex sets: for a set of vertices
$v_1,\lds,v_\el\in V(T)-\{v_r\}$ such that $e_i$ is the parent edge of $v_i$
in the rooted tree $T$, we define
\begin{gather}
  \cut(\{v_1,\lds,v_\el\}):=\cut(\{e_1,\lds,e_\el\}). \label{eq:cut-vtx}
\end{gather}
For every subset $U\s T_v$ of at most $k-1$ incomparable vertices and
integer $s\in[|U|,k-1]$, define $\DQ(v,s,U)$ to be the minimum value of
$\cut(U')$ over all subsets $U'\s T_v$ of size exactly $s$ whose
``maximal'' vertices are exactly $U$; in other words, $U'\supseteq U$
and every vertex in $U'-U$ is a descendant of (exactly) one vertex in $U$.

We now define a recursive statement for $\DQ(v,s,U)$. There are two cases, depending on whether $v\in U$ or
not. If $v\notin U$, then the following recursive statement is true:
\begin{gather}
  \DQ(v,s,U):= \min_{U',v_i,s_i,U'_i}\bigg(\sum_{i=1}^\el \DQ(v_i,s_i,U'_i) - \sum_{i=1}^\el
  \sum_{j=i+1}^\el w(E(T_{U'_i},T_{U'_j}))\bigg), \label{eq:DQDef}
\end{gather}
where the minimum is (i) over all $U'\supseteq U$ whose maximal vertices are exactly $U$, and $v_1,\lds,v_\el$ are the children of $v$ whose subtrees $T_{v_i}$ intersect $U'$, and $U'_i:= U'\cap T_{v_i}$; and (ii) over all positive integers $s_1,\lds,s_\el$ summing to $s$. Note that the weight in the double summation accounts for the double-counted
edges, and is thus subtracted from the expression. If $v\in U$, then the
recursion becomes
\begin{multline}
  \DQ(v,s,U) := \cut(\{v\}) + \min_{U',v_i,s_i,U'_i}\bigg(\sum_{i=1}^\el \lp \DQ(v_i,s_i,U'_i) -
  w(E(T_{U'_i},V-T_v)) \rp \\ - \sum_{i=1}^\el\sum_{j=i+1}^\el
  w(E(T_{U'_i},T_{U'_j})) \bigg), \label{eq:DQdef-vinU}
\end{multline}
where the minimum is (ii) over all positive integers $s_1,\lds,s_\el$ summing to $s-1$ this time, and with (i) the same. Again, all subtractions in the expression handle double-counted edges.
\begin{observation}
Starting with the base states
\begin{enumerate}
\item $\DQ(v,0,\emptyset)=0$,
\item $\DQ(v,1,\{v\})=\cut(\{v\})$, and
\item $\DQ(v,s,U)=\infty$ for $s\ge2$, $U\s\{v\}$ \end{enumerate}
for all leaves $v\in V(T)$, by applying DP with the recursions above, we can compute the correct values of $\DQ$.
  %The recursion above returns the correct values of $\DQ$.
\end{observation}

In order to compress the number of states for this dynamic program we
need the notion of \emph{important} nodes and \emph{representatives},
which we define in the next sections.
Given a subtree $T_v$, the important nodes $I_v$ should be thought of as
a constant-sized family of consistent ``samples'' of nodes, such that we
can ``round'' our guesses for which edges to delete to their nearest
sample points. These rounded set of states are, loosely speaking, the
representatives. 
 
% Since $T_v$ is a set of vertices in the subtree rooted at $v$,
% let $T[T_v]$ and $G[T_v]$ be the tree and graph edges induced by this
% set of vertices. A first attempt is as follows: for each set of
% $\el \le s$ incomparable\footnote{Like the definition of incomparable
%   edges in \S\ref{sec:incomp-rand}, a set of vertices of $V(T)$ is
%   incomparable if every root-leaf path in $T$ contains at most one
%   vertex in the set.} vertices $u_1,\lds,u_\el$ in $T_v$, store the best
% way to cut $s$ edges $e_1,\lds,e_s$ within $E(T[T_v])$ such that (a)~the
% parent edges of nodes $u_1, \ldots, u_\el$ are cut, and $s-\el$ other
% edges are cut in the subtrees below these nodes, and the cost of edges
% in $w(E(G[T_v])\cap \cut(\{e_1,\lds,e_s\}))$ is minimized. This last
% expression can be thought of as the cost of the optimal cut within the
% scope of $T_v$. Observe that these dynamic programming states capture
% enough information about $E_T[S_1^*,\lds,S_k^*]$, in the following
% sense: for the values of $s,\el,e_1,\lds,e_\el$ such that the component
% $S^*_i\in\m S^*$ containing vertex $v$ satisfies
% $S^*_i\cap T_v=T_v-\bigcup_i T_{u_i}$, the best way to cut has cost
% exactly $w(E(G[T_v])\cap OPT)$. Hence, this particular state is an
% accurate representation of $E_T[S_1^*,\lds,S_k^*]$ localized to
% $T_v$. 

%\paragraph{Representative states.}

\subsection{Important Nodes and Representatives}\label{sec:ImportantAndReps}

\subsubsection{Important Nodes}
\label{sec:Important}

Given tree $T$, some node $v \in V(T)$, we define a set of
\emph{important nodes} within $V(T_v)$, the nodes in the subtree below
$v$. One can think of these essentially as an ``$\e$-net'' of the nodes
in $T$, in a certain technical sense.
%In this section, we first specify their properties, and then show
% why such small important sets exist, and how to find them.
For each node $u\in T_{v}$, assign a vertex weight $\phi_v(u)$ to $u$
equal to the total weight of edges in $G$ that connect $u$ to vertices
outside $T_{v}$; i.e.,
\begin{gather}
  \phi_v(u) : = w(E_G(\{u\},V-T_{v})). \label{eq:weight}
\end{gather}
Observe that the total $\phi_v(\cdot)$ weight of all vertices in $T_{v}$
is exactly $E(T_{v},V-T_{v})$, which is at most $M$ by
Observation~\ref{obs:M}.  We want a set $I_{v} \sse V(T_{v})$ of
\emph{important nodes} for $T_v$ such that

\begin{leftbar}
  \begin{enumerate}
  \item[(P1)] Every connected component in $T_{v}-I_v$ has total
    $\phi_v(\cdot)$ weight at most $W:=\poly(\e/k)M$, whose exact value is determined later.

  \item[(P2)] The size $|I_v|$ of any important set is at most
    $4M/W + 1=\poly(k/\e)$.

  \item[(P3)] For each pair of vertices $v,p\in V(T)$ where
    $v$ is a descendant of $p$, $I_p\cap T_v\s I_v$.

  \item[(P4)] $v \in I_v$.
  \end{enumerate}
\end{leftbar}

\begin{theorem}[Important Nodes]
  \label{thm:algo-imp}
  Given $T$, there is a polynomial-time algorithm to find a set $I_v$ of
  important nodes for each node $v \in V(T)$, satisfying properties
  (P1)-(P4) above.
\end{theorem}

The approach is to start off with the empty set at the root, and proceed
top-down, adding nodes to representative sets as the $\phi$-weight
increases. The details, along with the proof of why the size remains controlled, are deferred to
\S\ref{sec:appendix-approx}, since they are orthogonal to the present story.

\subsubsection{Representative States}
\label{sec:repr-states}

The idea of representative states is simple: instead of keeping track of
all possibly $\approx n^k$ states in \DQ above, we ``round'' each state
to a close-by representative, such that there are only a constant number
of such representatives, but we incur only a small multiplicative error.

Fix a vertex $v\in V(T)$. We focus on computing representative states
for $v$ in the subtree $T_{v}$. Recall the notion of $\phi_v$-weight
from~(\ref{eq:weight}), and the properties~(P1)-(P3) of important nodes
$I_v \sse T_v$. 

%   \emph{\alert{Our first idea is a subset of \textit{important
%     vertices} $I\s T_{v}$, which is essentially an ``$\e$-net'' of the
%   vertices in the following sense: we want to approximate each true
%   state $(u_1,\lds,u_\el)$ with the \textit{representative} set $\lp
%   T_{v}-\bigcup_i T_{u_i}\rp \cap I$ to an $\e$-factor loss, a measure
%   we make precise later. For vertex $v$, the states are now indexed by
%   subsets of $I$. If $I$ has size FPT in $k$, then the number of states
%   is now $\le2^{|I|}$, which is also FPT in $k$.}}

% \paragraph{Important vertices.}

% For each node $u\in T_{v}$, assign a vertex weight $x_u$ to $u$ equal to
% $w(E_G(\{u\},V-T_{v}))$, which is the total weight of edges in $G$ that
% connect $u$ to vertices outside $T_{v}$. Observe that the total vertex
% weight of all vertices in $T_{v}$ is exactly $E(T_{v},V-T_{v})$, which
% we assumed is at most $M$.  We want to find a set $I\s T_{v}$ with the
% following two properties:
% \begin{OneLiners}
% \item[(1)] Every connected component in $T_{v}-I$ has total vertex
%   weight at most $W:=(\e/k^4)M$.
% \item[(2)] $|I|=\poly(M/W)=\poly(k/\e)$.
% \end{OneLiners}

% In \S\ref{sec:Important}, we show how to find such a set $I$. Here, we instead provide the necessary properties of important vertices that allow our dynamic programming to work.

\begin{definition}[Representatives]
  The \emph{representative} of a subset $S\s T_v$ is
  $\sigma_v(S):=S\cap I_v$, i.e., the set of important nodes within $S$.
  % has \defn{representative}
  % $\sigma_v(S)\s I_v$, defined by 
  %   subtrees}
\end{definition}

The function $\sigma_v:2^{T_v}\to 2^{I_v}$ maps sets to their
representatives. Typically, we will deal with representatives $R$ that
are downward-closed, that is, if $u \in R$ then all of its descendants
that are also in $I_v$ also belong to $R$. Given a representative, the
function $\sigma_v\inv$ is an inverse of sorts, indicating a canonical
set to consider for each representative. Figure~\ref{fig:round} gives a
pictorial depiction of this and following definitions.

\begin{definition}
Given a representative $R$, define its \emph{canonical set}
$\sigma_v\inv(R):= T_R := \cup_{x \in R} T_x$.
\end{definition}

Observe that $\sigma_v(\sigma_v\inv(R)) = R$ for any downwards-closed
representative $R$.  We now show that any set of the form
$S:=\bigcup_{i=1}^\el T_{u_i}$ is roughly equal to the canonical set of
its representative $\sigma_v\inv(\sigma_v(S))$, in the following sense:
$E(S,V-T_v) \approx E(\sigma_v\inv(\sigma_v(S)),V-T_v)$.
% Intuitively, this set of ``outside edges'' captures the essence of the
% state $S$: it is exactly the set of ``uncertain'' edges in
% $E[T_v,V-T_v]$, those that may or may not contribute to the final
% \kcut solution, depending on which edges in $T-T_v$ are cut.
For brevity, let us define
\begin{gather}
  \boxed{\round_v(S) := \sigma_v\inv(\sigma_v(S)).} \label{eq:round}
\end{gather}

\begin{figure}
  \centering
  \includegraphics[width=0.6\textwidth]{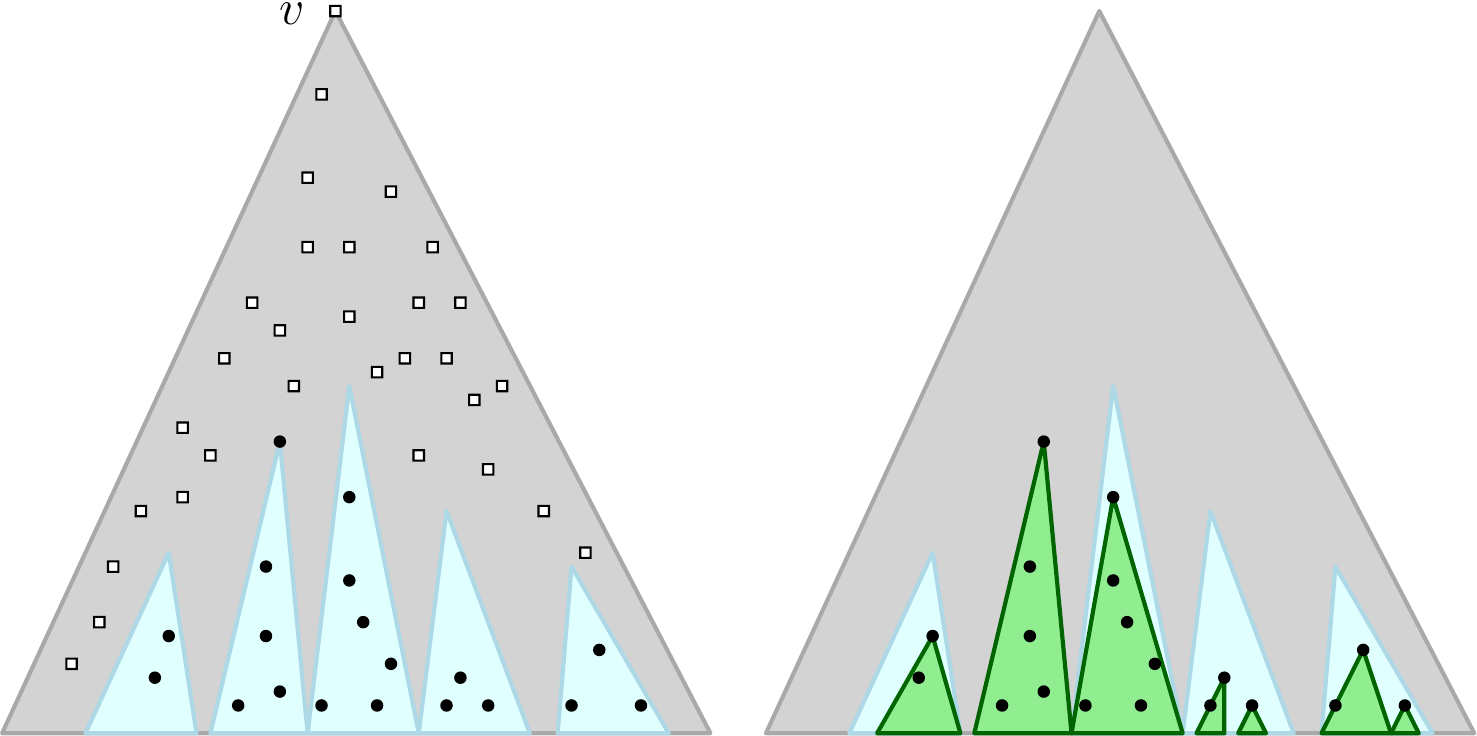}
  \caption{\label{fig:round} The boxes and dots are the set of important
    nodes $I_v$. For the downwards-closed blue set $S$, its
    representative $R = \sigma_v(S)$ is the set of important nodes
    (dots) within it. The union of the green subtrees $G$ then form the
    canonical set for $R$, and hence $\round_v(S) = G$.}
%\end{wrapfigure}
\end{figure}

\begin{lemma}\label{lem:ImportantApproximate}
  Fix incomparable vertices $u_1,\lds,u_\el\in T_{v}$, and let
  $S:=\bigcup_{i=1}^\el T_{u_i}$. Then,
\begin{OneLiners}
\item[(1)] $\round_v(S) \s S$
\item[(2)] $w(E(S-\round_v(S),V-T_v)) \le \el W$
\end{OneLiners}
\end{lemma}

\begin{proof}
  Property (1) clearly holds by the definitions of $\sigma_v$ and
  $\sigma_v\inv$. For property (2), suppose that $u\in
  S-\round_v(S)$. Recall that $T_v-I_v$ is a forest whose connected components each have weight $\le W$. 
        %every connected component in
  %$T_v-I_v$ has weight $\le W$. 
        Then, $u$ must be in a connected component of $T_v-I_v$
        containing a vertex $u_i$, for some $i\in[\el]$. Overall the vertices $u_1,\lds,u_\el$ are
  each responsible for one subtree of weight $\le W$. Thus,
  \[ w(E(S-\round_v(S),V-T_v)) = \sum_{u\in
    S-\round_v(S)}\phi_v(u) \le \el W . \qedhere\]
  % Fix a node $v\in S$. Recall that every connected tree in $T_s-I$ has
  % weight $\le W$. If $v\in I$, then $v\in S$ as well. Otherwise, $v$
  % lies in some subtree of $T_s-I$ of weight $\le W$. Note that every
  % node below this subtree is included in $T_{S'}$. Overall the $\le k$
  % nodes in $S$ are each responsible for one subtree of weight $\le W$;
  % every node below any of the subtrees is included in $T_{S'}$. So the
  % total weight unaccounted for is $\le kW$.
\end{proof}

Lemma~\ref{lem:ImportantApproximate} suggests that ``rounding'' each
state in \DQ to its representative should changes the answer only by
$kW$; we now proceed to the precise dynamic program.

\subsection{A Smaller Dynamic Program}

% After computing important nodes $I_v$ for each $v\in V$, we now employ
% dynamic programming.
% Before doing so, we first add vertex $v$ to each
% $I_v$; clearly, all of the properties of important nodes still hold,
% with $4M/W$ replaced by $4M/W+1$. 

Recall the ``ideal'' DP \DQ from \S\ref{sec:DQ}: we want to approximate
its states by a smaller set. The new DP \DPStar will still require $n^{\Om(k)}$
time to compute, but will closely approximate \DQ. In the next section
we show the final link in the chain: how to compute a good approximation
of \DPStar in FPT time.

We define $\DPStar(v,s,R)$ for each $v\in V(T)$, $s\in[k-1]$, and $R\s
I_v$, which we think of as the ``rounded'' DP states after approximating
sets with their representatives. While we do not always enforce it, imagine that $R\s I_v$ is always downward-closed inside $I_v$; that is, if $u\in R$, then $T_u\cap I_v\s R$. In a perfect world, the state
$\DPStar(v,s,R)$ would equal the smallest value of $\DQ(v,s,U)$ over all
$U$ such that $\sigma_v(T_U)=R$. At a high level, we will argue that
since representatives only cause a small additive error,
$\DPStar(v,s,R)$ will still be close to the smallest value of
$\DQ(v,s,U)$.

The base states for $\DPStar$ are the same: $\DPStar(v,0,\emptyset)=0$ and $\DPStar(v,1,\{v\})=\cut(\{v\})$ for leaves $v\in V(T)$.
Before we recursively define the DP states $\DPStar(v,s,R)$, we first
introduce the following lemma, which shows that the expressions for
$\DQ$ from~(\ref{eq:DQDef}) and~(\ref{eq:DQdef-vinU}) do not change much
if $T_{U_i}$ is replaced by $\round_{v_i}(T_{U_i})$, and similarly for
$T_{U_j}$. Note that we will be loose on the additive error bounds (such
as the bound $kW$ below), since it simplifies the argument and does not
affect our running time asymptotically.

\begin{lemma}\label{lem:DPAdditiveError}
  Consider a vertex $v\in V(T)$ and subset of incomparable vertices $U\s
  T_{v}$ of size at most $k$. Let $v_1,\lds,v_\el$ be all the children
  of $v$ whose subtrees $T_{v_i}$ intersect $U$, and define $U_i:=U\cap
  T_{v_i}$. Then,
  \[ 0\le
  w(E(T_{U_i},V-T_v))-w(E(\round_{v_i}(T_{U_i}),V-T_v))\le
  kW \] and
  \[0\le
  w(E(T_{U_i},T_{U_j}))-w(E(\round_{v_i}(T_{U_i}),\round_{v_j}(T_{U_j})))\le
  2kW .\]
\end{lemma}
\begin{proof}
  For the first inequalities, it suffices to show that $T_{U_i}\supseteq
  \round_{v_i}(T_{U_i})$ and
  \[ w(E(T_{U_i}-\round_{v_i}(T_{U_i}),V-T_v))\le kW .\] They follow
  immediately from Lemma~\ref{lem:ImportantApproximate} and the facts
  that $|U_i|\le k$ and $V-T_{v}\s V-T_{v_i}$.  For the second, observe
  that
  \begin{align*}
    &\ E(T_{U_i},T_{U_j})-E(\round_{v_i}(T_{U_i}),\round_{v_j}(T_{U_j}))  \\
    &~~~\s \bigg( E(T_{U_i}-\round_{v_i}(T_{U_i}),V-V_{v_i}) ) \cup (
    E(T_{U_j}-\round_{v_j}(T_{U_j}),V-V_{v_j}) \bigg) , 
  \end{align*}
  which has total weight at most $2kW$, again by Lemma~\ref{lem:ImportantApproximate}. 
\end{proof}

We now define the $\DPStar$ states. Recall that we have added vertex $v$
to $I_v$ for all $v$, so again, we have two cases. If $v\notin R$, then
the recursion is
\begin{gather}
  \DPStar(v,s,R) := \min_{\el,v_i,s_i,R_i} \lb \sum_{i=1}^\el
  \DPStar(v_i,s_i,R_i) - \sum_{i=1}^\el\sum_{j=i+1}^\el
  w(E(\sigma_v\inv(R_i), \sigma_v\inv(R_j))) \rb \label{eq:DPStarDef} ,
\end{gather}
where the minimum is over all choices of $\el\in[s]$, distinct children
$v_1,\lds,v_\el$ of $v$, positive integers $s_1,\lds,s_\el$ whose sum is
$s$, and representatives $R_1,\lds,R_\el$ of $v_i,\lds,v_\el$ such that
$R_i\s I_{v_i}$ and $R_i \cap I_v = R\cap T_{v_i}$ for each
$i\in[\el]$. For the last condition, $R_i \cap I_v = R\cap T_{v_i}$,
observe that if $I_{v_i} = I_v\cap T_{v_i}$, then $R_i$ must be $R\cap
T_{v_i}$. However,
$I_{v_i}\supseteq I_v\cap T_{v_i}$ in general, so we can view $R_i$ as a
``refinement'' of $R$ inside $T_{v_i}$.

If $v\in R$, then the $s_i$ satisfy $\sum_is_i=s-1$ instead, and the
recursion becomes
\begin{align*}
  \DPStar(v,s,R) := \cut(\{v\}) + \min_{\el,v_i,s_i,R_i} \bigg[
  &\sum_{i=1}^\el \lp
  \DPStar(v_i,s_i,R_i) - w(E(\sigma_v\inv(R_i),V-T_v)) \rp \\
  &- \sum_{i=1}^\el\sum_{j=i+1}^\el w(E(\sigma_v\inv(R_i),
  \sigma_v\inv(R_j))) \bigg];
\end{align*}

\begin{lemma}\label{lem:DPStarApprox}
  For each vertex $v\in V(T)$, integer $s\in[k-1]$, and downward-closed
  subset $R\s I_v$,
  \[ \min_{U:\sigma_v(T_U)=R} \DQ(v,s,U) \le \DPStar(v,s,R) \le
  \min_{U:\sigma_v(T_U)=R} \DQ(v,s,U) + (8s-4)k^2W .\]
\end{lemma}

\begin{proof}
  We apply induction from the leaves of the tree to the root. Observe
  that if $U\s T_v$ and $R\s I_v$ satisfy $\sigma_v(T_U)=R$, then $v\in
  U\iff v\in R$, so we can separate the cases $v\in R$ and $v\notin R$.

  \paragraph{Case 1: $v\notin R$.} To show the first inequality,
  consider the values $\el,v_i,s_i,R_i$ that achieve the minimum of
  $\DPStar(v,s,R)$ in (\ref{eq:DPStarDef}). By induction, for each
  $i\in[\el]$, there exists $U_i$ such that $\sigma_{v_i}(T_{U_i})=R_i$
  and $\DQ(v_i,s_i,U_i)\le \DPStar(v_i,s_i,R_i)$. Recalling the
  definition of $\round$, 
  \[ w(E(\sigma_v\inv(R_i),
  \sigma_v\inv(R_j)))=w(E(\round_{v_i}(T_{U_i}),\round_{v_j}(T_{U_j})))
  \ge w(E(T_{U_i},T_{U_j})), \] using Lemma~\ref{lem:DPAdditiveError}.
  Now matching the terms in the double summations of (\ref{eq:DQDef})
  and (\ref{eq:DPStarDef}) gives $\DQ(v,s,U) \le \DPStar(v,s,R)$.

  To show the second inequality, consider any $U$ such that
  $\sigma_v(T_U)=R$ and $\DQ(v,s,U)$ is defined. We first consider the
  case $\el=1$ in (\ref{eq:DQDef}): there is one child $v_1$ such that
  $U\s T_{v_1}$, then we also have $R\s T_{v_1}$, so by
  (\ref{eq:DQDef}), \[\DQ(v,s,U)=\DQ(v_1,s,U_1).\] Setting
  $R_1:=\sigma_{v_1}(T_U)$, we have ${R_1}\cap I_v=R=R\cap T_{v_1}$,
  where we used that $R$ is downward-closed.  By (\ref{eq:DPStarDef}),
  \[\DPStar(v,s,R)\le \DPStar(v_1,s,R_1),\]
  and by induction,
  \[\DPStar(v,s,R_1)\le \DQ(v,s,U_1),\]
  so putting the inequalities together gives $ \DPStar(v,s,R)\le
  \DQ(v,s,U) $.

  Now suppose that $\el>1$: let $v_1,\lds,v_\el$ be all the children of
  $v$ whose subtrees $T_{v_i}$ intersect $U$, and define $U_i:=U\cap
  T_{v_i}$ and $s_i:=|U_i|$. Again, we set $R_i:=\sigma_{v_i}(T_{U_i})$,
  which satisfies ${R_i}\cap I_v = R\cap T_{v_i}$. By induction, for
  each $i\in[\el]$,
  \[ \DPStar(v,s,R_i)\le \DQ(v,s,U_i) + (8s_i-4)k^2W. \] By
  Lemma~\ref{lem:DPAdditiveError}, the additive error of each of the
  terms in the double summations of (\ref{eq:DQDef}) and
  (\ref{eq:DPStarDef}) is at most $2kW$, and there are $\el\le k$ of
  them, incurring an additive error of at most $2k^2W$. Altogether, we
  have
  \begin{align*}
    \DPStar(v,s,R) &\le \sum_{i=1}^\el \DPStar(v_i,s_i,R_i) - \sum_{i=1}^\el\sum_{j=i+1}^\el w(E(\sigma_v\inv(R_i), \sigma_v\inv(R_j)))  \\
    &\le \sum_{i=1}^\el (\DQ(v_i,s_i,U_i)+(8s_i-4)k^2W) - \sum_{i=1}^\el \sum_{j=i+1}^\el w(E(T_{U_i},T_{U_j})) + 2k^2W \\
    & = \DQ(v,s,U) + \lp \sum_{i=1}^\el (8s_i-4) + 2 \rp k^2W \\
    &= \DQ(v,s,U) + (8s+ 2-4\el)k^2W.
  \end{align*}
  Since $\el\ge2$, we have $8s+ 2-4\el \le 8s-4$, proving $
  \DPStar(v,s,R) \le \DQ(v,s,U) + (8s-4)k^2W $.

  \paragraph{Case 2: $v\in R$.} Most of the arguments are similar, so
  they are omitted. We only show the proof for the case $U=\{v\}$, which
  implies $R=I_v$ and $v\in R$. In this case, we have
  $\DQ(v,s,U)=\DPStar(v,s,R)=\cut(\{v\})$, so they are equal.
\end{proof}

\subsection{A Dynamic Program in FPT Time}

In this section we compute a further approximation to $\DPStar(v,s,R)$,
called $\DPee(v,s,R)$, for each $v\in V(T)$, $s\in[k-1]$, and $R\s
I_v$.
The advantage of this approximation is that we can compute it in
time $(k/\e)^{O(k)}\poly(n)$. % We will call our algorithm states in order
% to differentiate it with the ``ideal'' states $\DPStar(v,s,R)$. 
Our main goal is to show that $\DPee(v,s,R)\approx \DPee(v,s,R')$ up to
a small additive factor.

\subsubsection{Partial Vertex Cover}\label{sec:pvc}

The base states for \DPee are the same as those in \DQ and \DPStar.
The computation of each recursive state, on the other hand, involves multiple calls to a
problem well-studied in the FPT setting, known as \emph{partial vertex
  cover}. We define a node-weighted version below.

\begin{definition}[(Node-Weighted) Partial $k$-Vertex Cover]
  Given a graph $G$ with node weights $\varphi:V\to[0,\infty)$ and edge
  weights $\psi:E\to[0,\infty)$, the \emph{node-weighted partial
    $k$-vertex cover} problem is to find a set $S\s V$ of exactly $k$
  nodes that minimizes $\varphi(S)+\psi(\bigcup_{v\in
    S}E(\{v\},V-\{v\}))$.
\end{definition}

The following theorem essentially follows from~\cite{GuptaLL18}, with an
improved running time from a more efficient coloring procedure and a
separate, trivial case when $\el=1$. Its proof is deferred to \S\ref{sec:appendix-approx}.
\begin{theorem}\label{thm:pvc}
  There is an $(1+\delta)$-approximation algorithm for node-weighted
  partial $\el$-vertex cover that runs in time
  $(k/\delta)^{O(k)}\poly(n)$. In addition, if $\el=1$, then the
  algorithm is optimal.
\end{theorem}

Recall from~(\ref{eq:DPStarDef}) that, ideally, we want
\[
\DPStar(v,s,R) := \min_{\el^\dag,v_i^\dag,s_i^\dag,R_i^\dag} \lb
\sum_{i=1}^\el \DPStar(v_i^\dag,s_i^\dag,R_i^\dag) -
\sum_{i=1}^\el\sum_{j=i+1}^\el w(E(\sigma_v\inv(R_i^\dag),
\sigma_v\inv(R_j^\dag))) \rb,
\]
over all choices of $\el^\dag\in[s]$, distinct children
$v_1^\dag,\lds,v_{\el^\dag}^\dag$ of $v$, positive integers
$s_1^\dag,\lds,s_{\el^\dag}^\dag$ whose sum is $s$, and representatives
$R_1^\dag,\lds,R_{\el^\dag}^\dag$ of $v_i^\dag,\lds,v_{\el^\dag}^\dag$
such that $R_i^\dag\s T_{v_i^\dag}$ and $T_{R_i^\dag} \cap I_v = R\cap
T_{v_i^\dag}$ for each $i\in[\el^\dag]$. Of course, iterating over all
possibilities may take $n^{\Om(k)}$ time, which is where we obtain the
speedup via partial vertex cover.

The intuition behind the partial vertex cover algorithm is as follows. For each child $u_i$ of $v$, suppose we have guessed a value $s_i\in[s]$ and $R_i\in T_{u_i}$. Construct a node-weighted complete graph $H'$ whose nodes are the children $u_i$ of $v$. Let the weight of each node $u_i\in V(H')$ be \[\varphi'(u_i):=\DPStar(u_i,s_i,R_i),\] and the weight of each edge $(u_i,u_j)\in E(H')$ be \[\psi'(u_i,u_j):=-w(E(\sigma_v\inv(R_i),\sigma_v\inv(R_j))).\] Observe how these weights relate to the expression in (\ref{eq:DPStarDef}). Now suppose that we got lucky: for each  child $v_i^\dag$ that achieves the minimum in (\ref{eq:DPStarDef}), we have guessed the correct corresponding values $s_i^\dag,R_i^\dag$ for that child, and moreover, we know the correct value $\el^\dag$. Then, if we run an algorithm that computes a node subset $S\s V(H'')$ of size exactly $\el^\dag$ that minimizes
\[ \sum_{u\in S}\varphi'(u)+\sum_{e\in E(H''[S])}\psi'(e) ,\]
then the optimal solution would return a set $S$ containing exactly the values $v_i^\dag$ minimizing (\ref{eq:DPStarDef}).

Of course, the above problem is not an instance of partial vertex cover, and even if it is, computing a solution exactly is $W[1]$-hard. To solve these issues, we will first transform the instance to one of partial vertex cover, and then run a $(1+\delta)$-approximation algorithm instead of an exact one.

\subsubsection{Defining \DPee}\label{sec:DefiningDP}

The instance for computing $\DPee(v,s,R)$ is as follows. First, there is
a guessing step, which is repeated multiple times. Let $\el$ be a
uniformly random value in $[s]$, and for each child $u$ of $v$, let
$s'(u)$ be a uniformly random value in $[s]$ and $R'(u)\s I_u$ be chosen
uniformly among those satisfying $R'(u)\s T_u$ and $T_{R'(u)}\cap I_v =
R\cap T_u$. We say that our guessing procedure succeeds if
$\el=\el^\dag$ and for each $v_i^\dag$, we have $s'(v_i^\dag)=s_i^\dag$
and $R'(v_i^\dag)=R_i^\dag$; we make no assumption on the children not
in $\{v_1^\dag,\lds,v_{\el^\dag}^\dag\}$. Clearly, we succeed with
probability at least
\[
\f1s \cd \lp \f1s \rp^\el \cd \lp \f1{2^{|T_{v_i^\dag}|}}\rp^\el = 2^{-\poly(k/\e)} .
\]
We will later repeat the procedure $2^{\poly(k/\e)}\log n$ times so that
w.h.p., we succeed at least once. For a given procedure, since there are
two cases depending on whether $v\in R$, we split into two cases.

\paragraph{Case 1: $v\notin R$.} 
We construct the auxiliary graph $H$ on which to compute partial vertex
cover. Let the children of $v$ be numbered $u_1,\lds,u_t$. For each
child $u_i$, add a node of weight
\begin{gather}
  \varphi(u_i):=\DPee(u_i,s'(u_i),R'(u_i)) - \sum_{j=1}^t
  w(E(\sigma_v\inv(R'(u_i)),\sigma_v\inv(R'(u_j))))+M. \label{eq:2}
\end{gather}
Note that
\[ \bigcup_{j=1}^t E(\sigma_v\inv(R'(u_i)),\sigma_v\inv(R'(u_j))) \s
E(T_{u_i}, V-T_{u_i}) ,\] 
so
\[ w(E(\sigma_v\inv(R'(u_i)),\sigma_v\inv(R'(u_j)))) \le w(E(T_{u_i},
V-T_{u_i}))\le M,\] 
and the node weight is always nonnegative.

For each pair $u_i,u_j$, $1\le i<j\le t$, we add an edge $(u_i,u_j)$ in $H$ of weight
\[ \psi(u_i,u_j):= w(E(\sigma_v\inv(R'(u_i)),\sigma_v\inv(R'(u_j)))) .\]

If the (exact) optimal solution to \pvc is $S$, then that solution has value
\begin{multline}
  \varphi(S)+\psi(\bigcup_{v\in S}E(\{v\},V-\{v\}))
  = \sum_{u\in S} \DPee(u,s'(u),R'(u))\\ + \el M 
  - \sum_{\{u,u'\}\s
    S}w(E(\sigma_v\inv(R'(u)),\sigma_v\inv(R'(u')))), \label{eq:PVCValue} 
\end{multline}
since that all edges $(u_i,u_j)$ such that $u_i\in S$, $u_j\notin S$ get
their weight canceled by the corresponding term in the negative
summation of $\varphi(u_i)$. Observe that, aside from the additive $\el
M$ and $\DPStar$ being replaced by $\DPee$, the solution value is
exactly the expression in the minimum from (\ref{eq:DPStarDef}) with
values $\el$ and $u,s'(u),R'(u)$ for $u\in S$. Thus, if the guessing is
successful, then the optimal value of \pvc is at most
\begin{gather}
  \DPee(v_i^\dag,s_i^\dag,R_i^\dag) - \sum_{i=1}^\el\sum_{j=i+1}^\el
  w(E(\sigma_v\inv(R_i^\dag), \sigma_v\inv(R_j^\dag))) + \el M
  . \label{eq:PVCBound} 
\end{gather}
We now run $(1+\delta)$-approximate partial $\el$-vertex cover on $H$,
for some $\delta:=\poly(\e/k)$, whose exact value is determined
later. Because of this approximation, we suffer a small loss.

The algorithm repeats the guessing and partial vertex cover computation
$2^{\poly(k/\e)}\log n$ times. On each iteration, the algorithm writes
down the value of the \pvc minus $\el M$, called the \textit{score} of
that iteration. Finally, the algorithm sets $\DPStar(v,s,R)$ as the
value of the best score found.

\paragraph{Case 2: $v\in R$.} 
Again, the nodes of $H$ consist of the children of $v$, numbered $u_1,\lds,u_t$. For each child $u_i$, 
\begin{align*}
  \varphi(u_i):=\DPee(u_i,s'(u_i),R'(u_i)) &- w(E(\sigma_v\inv(R_i),V-T_v)) \\
  & - \sum_{j=1}^t w(E(\sigma_v\inv(R'(u_i)),\sigma_v\inv(R'(u_j))))+2M,
\end{align*}
that is, the value~(\ref{eq:2}) from Case 1 with
$- w(E(\sigma_v\inv(R_i),V-T_v))+M$ added on; again, we can show that
the node weights are nonnegative. The edge weights $\psi(u_i,u_j)$ of
$H$ are the same as in Case 1. Similarly, the algorithm repeats the
procedure $2^{\poly(k/\e)}\log n$ times and on each iteration, writes
down the value of the \pvc plus $\cut(\{v\})-2\el M$.

\subsubsection{The Analysis}

The next lemma argues that in both cases above, w.h.p.,
$\DPee(v,s,R)\approx \DPStar(v,s,R)$.

\begin{lemma}\label{lem:ActualDPApprox}
  W.h.p., for each vertex $v\in V(T)$, integer $s\in[k-1]$, and
  downward-closed subset $R\s I_v$,  
  \[
  \DPStar(v,s,R) \le \DPee(v,s,R) \le \DPStar(v,s,R)+(8s-4)\delta kM.
  \]
\end{lemma}

\begin{proof}
  We apply induction from the leaves of the tree to the root. We only
  prove the case $v\notin R$, since the other case is almost identical.

  The first inequality essentially follows by induction and
  (\ref{eq:PVCValue}).  For the second inequality, we split into the
  cases $\el^\dag=1$ and $\el^\dag\ge2$.

  If $\el^\dag=1$, then the single $\DPee(v_1^\dag,s_1^\dag,R_1^\dag)$
  term introduces additive error $\le (8s-4)\delta kM$ by induction, and
  the \pvc algorithm outputs the optimal solution, so by
  (\ref{eq:PVCBound}), the score of a successful iteration is at most
  $\DPStar(v,s,R)+(8s-4)\delta kM$. Since the algorithm takes the best
  score over all iterations, the second inequality holds.

  If $\el^\dag \ge2$, then by (\ref{eq:PVCBound}), the score of a
  successful iteration is at most
  \[ (1+\delta)\lp \sum_{i=1}^{\el^\dag}
  \DPee(v_i^\dag,s_i^\dag,R_i^\dag) -
  \sum_{i=1}^{\el^\dag}\sum_{j=i+1}^{\el^\dag}
  w(E(\sigma_v\inv(R_i^\dag), \sigma_v\inv(R_j^\dag))) + \el M \rp - \el
  M.\] We may assume that $\DPee(v_i^\dag,s_i^\dag,R_i^\dag)\le M$,
  since anything larger will not result in a solution that beats $M$,
  the value of the $2$-approximation algorithm. Therefore, the score is
  upper bounded by
  \[
  \lp \sum_{i=1}^{\el^\dag} \DPee(v_i^\dag,s_i^\dag,R_i^\dag) -
  \sum_{i=1}^{\el^\dag}\sum_{j=i+1}^{\el^\dag}
  w(E(\sigma_v\inv(R_i^\dag), \sigma_v\inv(R_j^\dag))) \rp + \delta \el
  M + (1+\delta) \el M - \el M.
  \]
  By induction, \[\DPee(v_i^\dag,s_i^\dag,R_i^\dag) \le
  \DPStar(v_i^\dag,s_i^\dag,R_i^\dag) - (8s_i^\dag-4)\delta kM.\] Thus,
  \begin{align*}
    \DPee(v,s,R) 
    &\le  \sum_{i=1}^{\el^\dag} \DPee(v_i^\dag,s_i^\dag,R_i^\dag) -
    \sum_{i=1}^{\el^\dag}\sum_{j=i+1}^{\el^\dag}
    w(E(\sigma_v\inv(R_i^\dag), \sigma_v\inv(R_j^\dag)))  + 2\delta \el
    M \\ 
    &\le  \sum_{i=1}^{\el^\dag} \lp \DPStar(v_i^\dag,s_i^\dag,R_i^\dag) -
    (8s_i^\dag-4)\delta kM \rp -
    \sum_{i=1}^{\el^\dag}\sum_{j=i+1}^{\el^\dag}
    w(E(\sigma_v\inv(R_i^\dag), \sigma_v\inv(R_j^\dag)))  + 2\delta k M
    \\  
    &= \DPStar(v,s,R) - \sum_{i=1}^{\el^\dag} (8s_i^\dag-4)\delta kM +
    2\delta k M \\ 
    &= \DPStar(v,s,R) + (8s+ 2-4\el )\delta kM.
  \end{align*}
  Since $\el\ge2$, we have $8s+ 2-4\el \le 8s-4$, proving $\DPee(v,s,R)
  \le \DPStar(v,s,R)+(8s-4)\delta kM$. 
\end{proof}

\paragraph{Parameters.}
There are two free parameters, $W=\poly(\e/k)M$ from
\S\ref{sec:ImportantAndReps} and $\delta=\poly(\e/k)$ from
\S\ref{sec:pvc}. By Lemma~\ref{lem:DPStarApprox} and
Lemma~\ref{lem:ActualDPApprox},
\begin{align*}
  \DPee(v,s,R) &\le \DPStar(v,s,R)+(8s-4)\delta kM \\
  &\le \DQ(v,s,R) + (8s-4) \delta kM + (8s-4)k^2W \\
  &\le \DQ(v,s,R) + 8\delta k^2 M + 8k^3 W.
\end{align*}
Thus, setting $W:=\e/(32k^3)M$ and $\delta:= \e/(32k^2)$ gives
\[ \DPee(v,s,R) \le \DQ(v,s,R) + (\e/2)M \] for all $v\in V(T)$,
$s\in[k-1]$, $R\s T_v$. In particular, for the value $R^*\s T_{v_r}$
such that
\[\DQ(v_r,k-1,R^*)=w(OPT),\]
we have
\begin{align*}
  \DPee(v_r,k-1,R^*) &\le \DQ(v_r,k-1,R^*)+(\e/2)M \\
  &\le  \DQ(v_r,k-1,R^*)+ \e\cd w(OPT) \\
  &=(1+\e)w(OPT),
\end{align*}
where we have used that $w(OPT) \ge M/2$, since $M$ is the value of a
$2$-approximation algorithm. This concludes the $(1+\e)$-approximation
algorithm, which runs in time $2^{\poly(k/\e)}\poly(n)$.

With some more work, we can improve the runtime to
$(k/\e)^{O(k)}\poly(n)$ and make it deterministic; we defer the details to \S\ref{sec:s4-improvements} and \S\ref{sec:derandomize}, respectively.

%%% Local Variables:
%%% mode: latex
%%% TeX-master: "main"
%%% End:

%\input{approx-algo}

%\newpage

{\small \bibliographystyle{alpha}
\bibliography{refs}}

%\newpage
\appendix

%\newpage

\section{Lower Bounds}
\label{sec:lower-bounds}

We give the details of the simple relationship to the \kclique problem.

\Hardness*

\begin{proof}
  Given a graph $G = (V,E)$ that is an instance of \kclique, construct a
  graph $G'$ as follows. Take the graph $G$, add in a new vertex $s$
  with edges $(s,v)$ to all vertices $v \in V$, where the edge $(s,v)$
  has weight $n^2 - \text{deg}_G(v)$. It is easy to see that the optimal
  $(k+1)$-cut in the graph consists of $k$ parts containing singleton
  vertices $\{v_1\},\ldots, \{v_{k}\}$ from $V$, and one part containing
  $s$ and the rest of the vertices in $V$. Moreover, the weight of this
  $(k+1)$-cut is $n^2k - \frac12 E_G(\{v_1\},\ldots, \{v_{k}\})$. Hence
  the optimal $k$-cut will pick out a $k$-clique if one exists.
\end{proof}

In the above reduction, observe that a star rooted at the $s$ is a tight
\Ttree with respect to the optimal $k$-cut. Hence solving finding the
optimal $k$ edges to delete given this tight \Ttree find the $k$-clique.

%%% Local Variables:
%%% mode: latex
%%% TeX-master: "main"
%%% End:

\section{An $\apxfactor$-FPT Approximation Algorithm}
\label{sec:approx-fpt}

In this section, we show an $\apxfactor$-FPT approximation algorithm for
\kcut, proving Theorem~\ref{thm:fpt-apx}.  The best approximation factor
was $2 - \delta$ for $\delta \approx 0.0003$~\cite{GuptaLL18}.  Our
improvement is based on our $(1 + \eps)$-approximation algorithm given a
tight \Ttree. The tight tree was also used in~\cite{GuptaLL18} in a
special case of the $k$-cut called \Laminarkcut{k}, but the previous
algorithm only gave a $(2 - \delta)$-approximation for a small constant
$\delta > 0$. Simply plugging in our improved algorithm as a subroutine
and setting parameters more carefully gives the improvement, and also
simplifies the previous proof.  We now explain how we achieve a better
approximation for the general case, slightly modifying the parameters in
the proof of~\cite[Theorem~3.1]{GuptaLL18}.

Let $\{ S^*_1, \dots, S^*_k \}$ denote the partition of $V$ given by the
optimal $k$-cut, with $w(\partial S^*_1) \leq ... \leq w(\partial
S^*_k)$.  At a high level, the algorithm iteratively increases the
number of components by cutting the minimum cut or min-4-cut.  Let $k'$
be the current number of connected components and $S_1, \dots, S_{k'}$
be the components maintained by the algorithm.  In~\cite{GuptaLL18},
$\kl$ is defined to be the smallest value of $k'$ when both the weight
of the min-cut, as well as one-third of the weight of the min-4-cut,
becomes bigger than $w(\partial S_1^*)(1 - \eps_1 / 3)$, for some
$\eps_1 > 0$. Moreover, $\kr \in [k]$ is the smallest number such that
$w(\partial S^*_{\kr}) > w(\partial S_1^*)(1 + \eps_1 / 3)$. Let us
change these two hard-coded thresholds $(1 - \eps_1 / 3)$ and $(1 +
\eps_1 / 3)$ to $(1 - \alpha)$ and $(1+ \beta)$ respectively for
$\alpha, \beta > 0$ to be determined later. Let $S^*_{\geq \kr} =
\cup_{i = \kr}^{k} S^*$ be the union of the components with ``large''
boundary.

Consider the iteration for our algorithm when $G$ has been broken into
$\kl$ components $S_1, \dots, S_{\kl}$. By the choice of the threshold
$\kl$, both the weight of the mincut, and the weight of the min-4-cut
divided by $3$, are now bigger than $w(\partial S_1^*)(1 - \alpha)$. If
two cuts of weight at most $(1+\beta)w(\partial S^*_1)$ cross in $G$,
they will result a 4-cut in $G'$ of cost at most $2(1+\beta) w(\partial
S_1^*)$, which is contradiction to the choice of $\kl$ if
\begin{equation*}
  2(1+\beta) w(\partial S_1^*) < 3 (1 - \alpha) w(\partial S_1^*) \iff
  \frac{1 + \beta}{1 - \alpha} < 1.5. 
  \label{eq:para0}
\end{equation*}
Therefore, for each $S_i$, two cuts of size $(1+\beta)w(\partial S^*_1)$
also do not cross in $G[S_i]$. Now fix $i$ such that the component $S_i$
intersects at least two of $S^*_1, \dots, S^*_{\kr - 1}, S^*_{\geq
  \kr}$; say it intersects $r_i \geq 2$ of them. Now we would like to
take an $r_i$-cut within $G[S_i]$ with $\{ S^*_1 \cap S_i, \dots,
S^*_{\kr - 1} \cap S_i, S^*_{\geq \kr} \cap S_i \}$ as the desired
solution. Moreover, in this $r_i$-cut instance on $G[S_i]$, the laminar
structure of the cuts of weight $(1+\beta)w(\partial S^*_1)$ can be
encoded as a cut-tree (whose edges correspond to these non-crossing cuts
in $G[S_i]$, and we want to cut exactly $r_i - 1$ of them. This gives a
tight \Ttree, and we can use the algorithm from \S\ref{sec:approx} to
approximate the $r_i$-cut problem within $(1 + \eps_0)$-factor of the
desired solution. Indeed, we can do this for arbitrarily small $\eps_0 >
0$, not depending on any other parameter here.  This corresponds to the
\Laminarkcut{r_i} problem considered in~\cite{GuptaLL18}, but here we
have a much more relaxed requirement (i.e., existence of a tight tree).

Let us sketch the high-level idea of the rest of the proof for those who
don't remember details of~\cite{GuptaLL18}. In the paper, we assume that
the min-cut always remains smaller than $M := w(\partial S^*_1)$, else
we can branch on having found one component. The cost of the first $\kl$
cuts is $\kl (1-\alpha) M$, by the choice of $\kl$. Then guessing $r_i$
for each component $G[S_i]$ and running the FPT-PTAS for the tight
\Ttree instance gives us $(1+\eps_0) OPT$. Finally, we may have only
$\kl+\kr$ components, so we pick $(k - (\kl + \kr) - 1)$ other min-cuts,
each of cost at most $M$. Balancing the parameters now gives us the proof.

\paragraph{Technical details.} The rest of the analysis exactly works as
the original proof, where $(1 - \eps_1/3)$ is replaced by $(1 - \alpha)$
and $(1 + \eps_1/3)$ is replaced by $(1 + \beta)$.  (There is no
$\eps_1$ in the proof.)  The equations (1), (2), and (5)
in~\cite{GuptaLL18} that determine the parameters become
\[
2\alpha \eps_4 \geq \eps_3, \qquad (1 + \beta \eps_5)(2 - \eps_3) \geq
2, \qquad \eps_3 \leq 1 - 2 \eps_5,
\]
which is equivalent to
\[
\eps_3 = \min(2 \alpha \eps_4, \frac{2 \beta \eps_5}{1 + \beta \eps_5},
1 - 2 \eps_5).
\]
Setting $\alpha \approx 0.1588, \beta \approx 0.2618, \eps_4 \approx
0.5988$, and $\eps_5 \approx 0.4012$ gives $\eps_3 \approx 0.1901$,
which gives us an $(2 - \eps_3) \approx \apxfactor$-approximation.

As for running time, there is the same $2^{O(k^2)}\poly(n)$ multiplicative overhead in the reduction to \Laminarkcut{k} in~\cite{GuptaLL18}, which is the dominant factor in the overall runtime.

%%% Local Variables:
%%% mode: latex
%%% TeX-master: "main"
%%% End:

\section{Time and Space Requirements for Section~\ref{sec:randomized}}
\label{sec:app-random}

In this section, we show how to improve the runtime for the algorithm in
Section~\ref{sec:randomized}, to complete the proof of
Theorem~\ref{thm:FasterKCut}. We then talk about a bounded-space algorithm.

\subsection{Improvements to the Runtime}

First, we explain how to replace the factor of $m$ in the running time
with a potentially smaller factor of $n$. A closer examination of the
weights $w_H(v^{F^a}_a,v^{F^b}_b)$ shows that they are all nonnegative,
and that any term $w(\st(e^a_i,\sigma(R_a+i)))$
used to compute an edge $w_H(v^{F^a}_a,v^{F^b}_b)$ also lower bounds the
weight of any triangle containing that edge. Moreover, the minimum \kcut
has value $\le knW$, since isolating $k-1$ vertices is always a valid
\kcut. Therefore, in each graph $H$ that we construct, we can ignore any
edge with weight $>knW$, since they can only result in solutions with
value $>knW$. Now that the weights are in the range $[0,knW]$, we can
apply~\cite{williams2010subcubic} with $M:=knW$ to obtain the desired
running time. 

Another source of improvement occurs when $k\ne2\bmod3$, giving some
slack from the ceiling in $\lc(k-2)/3\rc$ when computing the states $\st(e,s)$. Let $r:=k-2\bmod3$; note that
$r \in \{0, 1,2\}$. In this case, it is more beneficial to guess $r$
edges to delete using brute force, and then apply the algorithm of
\S\ref{sec:general-rand} on the remaining $k-r$ edges to delete. Since
$T$ is a tight \Ttree, we claim that there exist $r$ edges in
$E_T(S_1^*,\lds,S_k^*)$ such that if they are removed from $T$, then $r$
of the $r+1$ connected components are exactly equal to some $r$ elements
in $\m S^*=\{S_1^*,\lds,S_k^*\}$.  Indeed, consider the process of
rooting the tight \Ttree $T$ at an arbitrary vertex and, for $r$
iterations, removing an edge in $E_T(S_1^*,\lds,S_k^*)$ of maximal
depth. In each iteration, since the subtree below the removed edge has
no more edges in $E_T(S_1^*,\lds,S_k^*)$, it must be an element in
$\m S^*$. Note that this process is not part of our algorithm; we
provide it only to prove existence.

The algorithm tries all $O(n^r)$ edges to remove, and for each one,
guesses which $r$ of the $r+1$ connected components are in $\m S^*$. If
we guess everything correctly, then we can run the algorithm of
\S\ref{sec:general-rand} to delete the other $k-2-r$ edges from the last
component $T'$.
% to try all possible ways to delete $r$ edges
%in $T$, treat $r$ of the resulting $r+1$ components as completed, and
%find the best way to cut $k-1-r$ edges inside the last component
%$T'$. 
That is, the input graph now becomes $G[V(T')]$ and the tree $T'$.

With these two improvements, the new running time for each $\st(e,s)$ computation becomes $\tO(k^{O(k)}n^{\lf(k-2)/3\rf\om+1+(k-2)\bmod3}W)$.
Similarly, by setting $r:=k-1\bmod3$, the final \kcut value can be computed in $\tO(k^{O(k)}n^{\lf(k-1)/3\rf\om+1+(k-1)\bmod3}W)$ time. Again, this is dominated by the running time $\tO(k^{O(k)}n^{\lf(k-2)/3\rf\om+2+(k-2)\bmod3}W)$ of computing all the states, attaining the bound in
Lemma~\ref{lem:Remove}.

\subsection{A Polynomial Space Algorithm}

The second improvement idea also leads to a polynomial space
algorithm. For a given constant $c$, apply the idea with $r:=k-1-c$, so
that the algorithm takes space $n^{O(k-1-r)}=n^{O(c)}$ and time
\[\tO(k^{O(k)}n^{\lc(k-1-r)/3\rc\om+1+r}W)=\tO(k^{O(k)}n^{(\om/3)c+2+(k-1-c)})=\tO(k^{O(k)}n^{k+1-0.2c}),\]
using $\om<2.3727$. Since there are $\tO(k^{O(k)}n^{k-1})$ trees to
consider by Lemma~\ref{lem:Merge}, the total running time is
$\tO(k^{O(k)}n^{2k-0.2c})$. 

% \begin{theorem}[Bounded-Space Algorithm]
%   \label{cor:PolynomialSpace}
%   Let $c$ be any constant satisfying $c=O(k)$. There is a randomized
%   algorithm for exact \kcut on unweighted graphs in time $O(n^{2k-c})$
%   and space $n^{O(1)}$.
% \end{theorem}

%%% Local Variables: 
%%% mode: latex
%%% TeX-master: "main"
%%% End: 

\section{Proofs from Section~\ref{sec:approx}}
\label{sec:appendix-approx}

\begin{proof}[Proof of Theorem~\ref{thm:algo-imp}]
  The algorithm proceeds top-down, starting with constructing $I_{v_r}$
  for the root $v_r$ and going downwards. For root $v_r$, the singleton
  $I_{v_r} = \{v_r\}$ satisfies constraints~(P1) and~(P2). Now we
  proceed top-down in the tree.

  Consider a child $v$ with parent $p$. Having already defined $I_p$ we
  start off with the set $I_p\cap T_v$ as a candidate for
  $I_v$. However, this may not satisfy~(P1), since the $\phi_v$-weight
  of a node $u\in T_v$ can be higher than its $\phi_p$-weight (but not
  lower), so the $\phi_v$-weight of a component in $T_v-(I_p\cap T_v)$
  may exceed $W$. We fix it as follows: for each component $C$ in
  $T_v-(I_p\cap T_v)$ has $\phi_v$-weight more than $W$, we run the
  following greedy bottom-up algorithm inside that component, producing
  additional important nodes.

  The algorithm is the natural one: we greedily pick the lowest vertex
  $u$ in $C$ with subtree $\phi_v$-weight more than $W/2$, mark it as
  important, remove its subtree, and repeat until the remainder has
  $\phi_v$-weight at most $W/2$. More formally, view the component $C$
  as a tree with the same ancestor-descendant relationship as in
  $T$. For any node $x \in V(C)$ let $C_x$ be the subtree of $C$ rooted
  at node $x$. The greedy algorithm maintains a set $I$ of newly picked
  important nodes in $C$, and iteratively adds to $I$ the node
  $u\in V(C) - \bigcup_{x\in I} V(C_x)$ of maximal depth that satisfies
  $\phi_v(V(C_u) - \bigcup_{x\in I}V(C_x))\ge W/2$, until such a vertex
  no longer exists. It is clear that every connected component in $C-I$
  has total $\phi_v$-weight at most $W/2$. %Now we

  % Since the total $\phi_v$-weight of nodes in $T_v$ is at most $M$, and
  % every time we add a node in $I$, we by at least $W$ every time a
  % vertex is added to $I$, so $|I|\le M/W$ at the end.  Now define
  % $I_v := I_p \cup I$.

  Let $I_v$ be $I_p \cap T_v$, plus these newly chosen important
  nodes. By construction, each subtree in $T_v - I_v$ has
  $\phi_v$-weight at most $W/2$; this satisfies property~(P1). We prove
  property~(P2) next.

  % We claim that $I_v$ has small size for
  % all $v$, by Combining these two claims and replacing each occurrence
  % of $W$ in the algorithm with $W/2$ fulfills all three properties.

  \begin{lemma} For all $v\in V$, $|I_v|\le 4M/W$. 
  \end{lemma}
  \begin{subproof}
    For any vertex $v\in V$, the important node $b \in I_v$ is \emph{in
      charge} of component $C$, if $C$ contains some child of $b$. The
    important node $b$ is \emph{active} (w.r.t.\ $v$) if the total
    $\phi_v$-weight of $b$, unioned with the components $C$ it is in
    charge of, is at least $W/2$. Otherwise $b$ is called \emph{retired}
    (w.r.t.\ $v$). There are $\le 2M/W$ active vertices, since each one
    is in charge of a disjoint set of components of weight $\ge W/2$ and
    the total weight of $T_v$ is $\le M$. We now bound the number of
    retired vertices.

    Consider the highest ancestor $u$ of node $v$ such that $b \in I_u$.
    Let $b$ be in charge (with respect to $u$) of components
    $C_1, C_2, \ldots, C_j$. By construction the total weight
    $\phi_u(b) + \sum_{j' \leq j} \phi_u(C_{j'}) \geq W/2$, and hence
    $u$ lies \emph{strictly} above $v$. Hence $b$ is an active important
    node with respect to this ancestor $u$. As we walk down the $u$-$v$
    path (with the $\phi$-weights non-decreasing), consider the first
    node $a$ such that when building the important set $I_{a}$, some
    vertex in $\cup_{j' \leq j} C_{j'}$ (say in $C_i$) is chosen as an
    important node. (Clearly $a$ is either $v$ or an ancestor of $v$.)
    At this point the $\phi_a$-weight in $C_i$ must have increased to
    $W$, because of new edges from nodes in $C_i$ to ancestors of $a$,
    of total edge-weight $\geq W/2$. Now we can ``charge'' the
    retirement of $b$ to these edges. It is clear edges are charged this
    way only by the important node $b$ which happened to be in charge of
    the current component they are incident to. Moreover, the total
    weight of such edges is at most $M$, since they all go from within
    $T_v$ to outside it; so the number of retired nodes is also at most
    $2M/W$.
  \end{subproof}
  The containment property~(P3) is true by construction. To get
  property~(P4) we simply add $\{v\}$ to $I_v$, for each $v$. This
  increases the size by $1$, and completes the proof.
\end{proof}

\begin{proof}[Proof of Theorem~\ref{thm:pvc}]
  We follow the same strategy as Theorem 5.1 in~\cite{GuptaLL18}, except
  with slightly different coloring probabilities. In their Lemma 5.2,
  instead of coloring each node red and blue with probability $\frac12$
  each, we color them red with probability $1/\tau$ and blue with
  probability $(1-1/\tau)$, for $\tau:=\poly(k/\delta)$ as defined
  there. This way, following their definition, the probability that all
  the nodes in $S^*$ are colored red, and all the nodes in
  $N(S^*)\setminus S$ are colored blue is $(1/\tau)^k(1-1/\tau)^\tau =
  (\delta/k)^{O(k)}$. We thus repeat this step $(k/\delta)^{O(k)}\log n$
  times, giving the desired running time.

  Finally, the case $\el=1$ can be trivially solved optimally, since the
  solution is simply the minimum weight node.

We remark that this algorithm is derandomized in \S\ref{sec:derandomize}.
\end{proof}

\section{Running Time Improvements for Section~\ref{sec:approx}}
\label{sec:s4-improvements}

Here, we show that the running time of the dynamic program from
\S\ref{sec:DefiningDP} can be sped up to $(k/\e)^{O(k)}\poly(n)$. The main idea is that for each vertex $v\in V(T)$, there are only $(k/\e)^{O(k)}$ many (downward-closed) subsets $R\s I_v$ that need to be considered, which is much smaller than the trivial $2^{\poly(k/\e)}$ bound as stated before.

For a given $v\in V(T)$ and $s\in[k-1]$, we say that a representative $R\s I_v$ is \emph{$(v,s)$-relevant} if there exists a set $U\s T_v$ of  $s$ incomparable vertices such that $\sigma_v(T_U)=R$. Intuitively, the only values of $\DPStar(v,s,R)$ that ``matter'' are the ones where $R$ is $(v,s)$-relevant. Formally, it can be shown, by analyzing the recursive definition of $\DPStar$, that $\DPStar(v,s,R)<\infty$ if and only if $R$ is $(v,s)$-relevant; here, we assume that $\DPStar(v,s,R)$ becomes $\infty$ if there do not exist $\el,v_i,s_i,R_i$ in (\ref{eq:DPStarDef}) that satisfy the necessary constraints, or if every satisfying $\el,v_i,s_i,R_i$ has $\DPStar(v_i,s_i,R_i)=\infty$ for some $i\in[\el]$. It follows that in our DP algorithm, we only need to compute $\DPee(v,s,R)$ for $(v,s)$-relevant $R$.

Below, we will prove that there are $(k/\e)^{O(k)}$ many representatives that are $(v,s)$-relevant, and we can enumerate them, plus possibly some more representatives that are not $(v,s)$-relevant, in $(k/\e)^{O(k)}$ time. Therefore, the DP algorithm can perform this enumeration and compute $\DPee(v,s,R)$ for only these $R$.  Moreover, in the guessing step in \S\ref{sec:DefiningDP},
for each child $u$ of $v$, we only need to choose a random $R'(u)\s I_u$ that is relevant in $T_u$, so the success probability increases to $(k/\e)^{O(k)}$. Overall, the running time of the DP algorithm becomes $(k/\e)^{O(k)}\poly(n)$.

\begin{lemma}
For a fixed vertex $v\in V(T)$ and integer $s\in[k-1]$, there are $(k/\e)^{O(s)}$ many $(v,s)$-relevant representatives, and we can enumerate a superset of all $(v,s)$-relevant representatives in $(k/\e)^{O(s)}$ time.
\end{lemma}

\begin{proof}

We first prove the statement when $s=1$. We use the concept of VC dimension, defined below.

\begin{definition}
Let $X$ be a set of elements, called the universe. A family $\m F$ of subsets of $X$ has VC dimension $d$ if $d$ is the largest possible size of a subset $S\s X$ satisfying the following property: for any subset $S'\s S$, there exists subset $F\in \m F$ such that $S\cap F=S'$.
\end{definition}
We use two properties of VC dimension. The first is that if $\m F$ if a family of subsets of $X$ of VC dimension $d$ and $Y\s X$, then the family $\m F\vert_Y:=\{F\cap Y:F\in\m F\}$ has VC dimension $\le d$.
The second property is a classic result on VC dimension:
\begin{theorem}[Sauer's lemma]\label{thm:Sauer}
Let $X$ be a set of elements. If a family $\m F$ of subsets of $X$ of VC dimension $d$, then $|\m F|=O(|X|^d)$.
\end{theorem}

We now bound the VC dimension of the family of all subtrees.
\begin{claim}
For a fixed vertex $v\in V(T)$, let $T_v$ be the universe. The family $\m F$ of subtrees $T_u$ for all $u\in T_v$ has VC dimension at most $2$.
\end{claim}
\begin{subproof}
Suppose for contradiction that there is a set $S\s T_v$ of size $3$ such that for each subset $S'\s S$, there exists a subtree $T_u\s T_v$ such that $S\cap T_u=S'$. Let $S=\{x,y,z\}$, and assume without loss of generality that the lowest common ancestor of $x$ and $y$ is either equal to or a descendant of the lowest common ancestor of $x$ and $z$. Then, if a subtree $T_u$ contains both $x$ and $z$, then it must contain $y$, so for the subset $S':=\{x,z\}$, it is impossible that $S\cap T_u=S'$, a contradiction.
\end{subproof}

By Theorem~\ref{thm:Sauer}, the family $\m F_v:=\{T_u:u\in T_v\}$ has VC dimension $\le2$. Therefore, the family $\m F_v\vert_{I_v}$ of subsets of $I_v$, which is precisely the set of $(v,1)$-relevant representatives, also has VC dimension $\le2$. Thus, there are $O(|I_v|^2)=(k/\e)^{O(1)}$ many $(v,1)$-relevant representatives, proving the case $s=1$. Moreover, we can enumerate over all of them in $(k/\e)^{O(1)}\poly(n)$ time.

For general $s$, let $U\s V$ be a subset $s$ incomparable vertices. Since $T_U=\bigcup_{u\in U}T_u$ is a union of $s$ subsets in $\m F_v$, it follows that $T_U\cap I_v$ is a union of $s$ subsets of $\m F_v\vert_{I_v}$. Since there are at most $(\m F_v\vert_{I_v})^s$ many possible such unions, the number of $(v,s)$-relevant representatives is $(k/\e)^{O(s)}$. Furthermore, to enumerate a superset of them, we can first compute $\m F_v\vert_{I_v}$ and then enumerate over all unions of $s$ subsets, taking $(k/\e)^{O(s)}\poly(n)$ time. This concludes the proof.
\end{proof}
With this speedup, our running time matches the one promised by Theorem~\ref{thm:kcut-approx}.

\section{Derandomization}\label{sec:derandomize}

The guessing part can be derandomized in the same way randomized FPT algorithms are typically derandomized: through efficient constructions of set families. The main impact of derandomization is the deterministic runtime of Theorem~\ref{lem:fpt-approx-tight}, which itself leads to the deterministic runtime of Theorem~\ref{thm:fpt-apx}.

We first derandomize the occasions when the algorithm has to guess multiple values in the range $[q]$ for some $q:=\poly(k/\e)$. More precisely, the algorithm guesses a value in $[q]$ for each index $i\in I$, such that for an unknown set of indices $I^*\s I$ of size $\le k$, we must guess the value correctly for each index $i\in I^*$. This occurs during the matrix multiplication algorithm in \S\ref{sec:general-rand} and the computation of $\DPee$ in \S\ref{sec:DefiningDP}. We derandomize this procedure using $(n,k,q)$-universal sets as introduced in~\cite{misra2013parameterized}.

\begin{definition}[Definition 3.1 of~\cite{misra2013parameterized}] An $(n,k,q)$-universal set is a set of vectors $V\s[q]^n$ such that for any index set $S\in\bn{[n]}k$, the projection of $V$ on $S$ contains all possible $q^k$ configurations.
\end{definition}
Note that the traditional notion of $(n,k)$-universal sets is precisely  the $(n,k,2)$-universal sets.

\begin{lemma}[Theorem 3.2 of~\cite{misra2013parameterized}]
 An $(n,k,q)$-universal set
 of cardinality $q^kk^{O(\log k)}\log^2n$ can be constructed deterministically in time $O(q^kk^{O(\log k)}n\log^2n)$.
\end{lemma}
Therefore, we can construct an $(n,k,q)$-universal set in time $O(q^kk^{O(\log k)}n\log^2n) = O((k/\e)^{O(k)}\poly(n))$ and run the inner procedure on each element in the set.

We now derandomize the \pvc algorithm, making the entire algorithm of Lemma~\ref{lem:fpt-approx-tight} deterministic.
 To do so, we use the following special construction of set families:
\begin{lemma}[Lemma I.1 of~\cite{Chitnis}]
Given a set $U$ of size $n$, and integers $0\le a,b\le n$, one can in (deterministic) $O(2^{O(\min(a,b))\log(a+b))}n\log n)$ time construct a family $\m F$ of at most $O(2^{O(\min(a,b)\log(a+b))}\log n)$ subsets of $U$, such that the following holds: for any sets $A,B\s U$, $A\cap B=\emptyset$, $|A|\le a$, $|B|\le b$, there exists a set $S\in\m F$ with $A\s S$ and $B\cap S=\emptyset$.
\end{lemma}

Following the proof of Theorem~\ref{thm:pvc} in \S\ref{sec:appendix-approx}, we set $U$ to be the nodes in the \pvc instance, and parameters $a:=k$ and $b:=\tau=\poly(k/\e)$. We construct a set family $\m F$ of size $O(2^{O(k\log(\poly(k/\e)))}\log n)=(k/\e)^{O(k)}\log n$ such that there exists a set $F\in\m F$ with $S^*\s F$ and $(N(S^*)\setminus S)\cap F=\emptyset$. Therefore, for each set $F\in\m F$, we color all nodes in $F$ red and all other nodes blue, and proceed with the algorithm.

\end{document}